\documentclass[draft,amssymb,12pt,oldlfont,reqno]{amsart}
\usepackage{amsmath,oldlfont,amssymb}

\normalsize
\def         \ad       {\text{adj}}

\numberwithin{equation}{section}

\newtheorem{thm}                    {Theorem}
\newtheorem{lem}[equation]      {Lemma}
\newtheorem{prop}[equation]     {Proposition}
\newtheorem{cor}[equation]          {Corollary}

\newtheorem{dfn}{Definition}
\def\derr#1{ \frac{d\,{#1}}{d\,r}}
\def\gL{{\ifmmode \Lambda \else $\Lambda $ \fi }}
\def\gs{{\ifmmode \sigma \else $\sigma $ \fi }}
\def\gk{{\ifmmode \kappa \else $\kappa $ \fi }}
\def\gl{{\ifmmode \lambda \else $\lambda $ \fi }}
\def\gm{{\ifmmode \mu \else $\mu $ \fi }}
\def\df{\stackrel{\rm def}{=}}
\def\cA{{\ifmmode \cal A \else $\cal A $ \fi }}
\def\cG{{\ifmmode \cal G \else $\cal G $ \fi }}
\def\cC{{\ifmmode \cal C \else $\cal C $ \fi }}
\def\gz{{\ifmmode \zeta \else $\zeta $ \fi }}
\def\cL{{\ifmmode \cal L \else $\cal L $ \fi }}
\def\cF{{\ifmmode \cal F \else $\cal F $ \fi }}
\def\cM{{\ifmmode \cal M \else $\cal M $ \fi }}
\def\fL{{\ifmmode \frak L \else $\frak L$ \fi}}
\def\gth{{\ifmmode \theta \else $\theta $ \fi }}
\def\fS{{\ifmmode \frak S \else $\frak S$ \fi}}
\def\cW{{\ifmmode \cal W \else $\cal W $ \fi }}
\def\gz{{\ifmmode \zeta \else $\zeta $ \fi }}
\def\fz{{\ifmmode \frak z \else $\frak z$ \fi}}
\def\cZ{{\ifmmode \cal Z \else $\cal Z $ \fi }}
\def\fF{{\ifmmode \frak F \else $\frak F$ \fi}}
\def\cH{{\ifmmode \cal H \else $\cal H $ \fi }}
\def\fD{{\ifmmode \frak D \else $\frak D$ \fi}}
\def\BH{{\ifmmode \Bbb H\else $\Bbb H$ \fi }}
\def\gep{{\ifmmode \epsilon \else $\epsilon $ \fi }}
\def\cI{{\ifmmode \cal I \else $\cal I $ \fi }}
\def\cJ{{\ifmmode \cal J \else $\cal J $ \fi }}
\def\fJ{{\ifmmode \frak J \else $\frak J$ \fi}}
\def\gn{{\ifmmode \nu \else $\nu $ \fi }}
\def\fN{{\ifmmode \frak N \else $\frak N$ \fi}}
\def\cP{{\ifmmode \cal P \else $\cal P $ \fi }}
\def\fM{{\ifmmode \frak M \else $\frak M$ \fi}}
\def\fI{{\ifmmode \frak I \else $\frak I$ \fi}}
\def\fX{{\ifmmode \frak X \else $\frak X$ \fi}}
\def\fW{{\ifmmode \frak W \else $\frak W$ \fi}}
\def\gb{{\ifmmode \beta \else $\beta $ \fi }}
\def\fc{{\ifmmode \frak c \else $\frak c$ \fi}}
\def\cX{{\ifmmode \cal X \else $\cal X $ \fi }}
\def\cD{{\ifmmode \cal D \else $\cal D $ \fi }}
\def\gt{{\ifmmode \tau \else $\tau $ \fi }}
\def\ga{{\ifmmode \alpha \else $\alpha$ \fi }}
\def\BR{{\ifmmode \Bbb R\else $\Bbb R$ \fi }}
\def\fV{{\ifmmode \frak V \else $\frak V$ \fi}}
\def\gD{{\ifmmode \Delta \else $\Delta $ \fi }}
\def\fR{{\ifmmode \frak R \else $\frak R$ \fi}}
\def\fC{{\ifmmode \frak C \else $\frak C$ \fi}}
\def\fA{{\ifmmode \frak A \else $\frak A$ \fi}}
\def\fs{{\ifmmode \frak s \else $\frak s$ \fi}}
\def\fw{{\ifmmode \frak w \else $\frak w$ \fi}}
\def\cB{{\ifmmode \cal B \else $\cal B $ \fi }}
\def\fm{{\ifmmode \frak m \else $\frak m$ \fi}}
\def\cQ{{\ifmmode \cal Q \else $\cal Q $ \fi }}
\def\cY{{\ifmmode \cal Y \else $\cal Y $ \fi }}
\def\fQ{{\ifmmode \frak Q \else $\frak Q$ \fi}}
\def\fO{{\ifmmode \frak O \else $\frak O$ \fi}}
\def\fa{{\ifmmode \frak a \else $\frak a$ \fi}}\def\gd{{\ifmmode \delta \else $\delta $ \fi }}
\begin{document}
\begin{titlepage}
\title{On the Cowling Approximation: A Verification of Ansatz via Methods of Functional and Asymptotic Analysis}
\author{ Christopher J. Winfield}
\address{Department of Mathematics and Statistics\
Oakland University\\
Rochester, MI
} \email{} \pagestyle{myheadings}
\markboth{Draft: C. Winfield}{Report}

\end{titlepage}

\begin{abstract}
We study the Cowling approximation by analytical means as applied to a system of linear differential equations arising from models of non-radial stellar pulsation. 
We consider various asymptotic cases, including those of high harmonic degree and high oscillation frequency.
Our methods involve a reformulation of the system in terms of an integro-differential equation 
for which certain Hilbert-space methods apply. By way of a more complete asymptotic study, we extend our results to certain fundamental solution sets,  characterized according to certain multi-point boundary-value problems: Such asymptotics further enable us to produce sharp estimates as confirmation of our general results.

\begin{sloppypar}
  \end{sloppypar}
\end{abstract}
\maketitle
\section{Introduction}\label{introsec}
We study approximation methods applied to the system
\begin{align} \derr{u}&=\frac{g}{c^2}u+\left[\gL_{\ell}-\gs^2\frac{r^2}{c^2}\right]\eta +\frac{r^2}{c^2}\Phi \label{otherfirst}\\
\derr{\eta}&=\frac{1}{r^2}\left(1-\frac{N^2}{\gs^2}\right)u+\frac{N^2}{g}\eta -\frac{N^2}{\gs^2g}\Phi\label{othersecond}\\
\frac{1}{r^2}\derr{\phantom{r}}\left(r^2\derr{\Phi}\right)&-\left[\frac{\gL_{\ell}}{r^2}-\frac{\gk\rho}{c^2}\right]\Phi =\gk\rho\left(\frac{N^2}{r^2g}u+\frac{\gs^2}{c^2}\eta\right)\label{otherthird}
\end{align}
(Eisenfeld and Smeyers - ES) (See \cite{svh} and see \cite{acdk,co,lw,uoas} for a general introduction to the theory of stellar pulsation.)
which has smooth, real-valued coefficients on intervals $0<a\leq r\leq b<\infty$ and parameters $\gL_{\ell},\gs>0.$   
Here dependent variables $\eta$ and $u$ arise from scalar potentials of Lagrangian displacement $\delta_{_L}\vec{r}$ and
$\Phi$ by the Eulerian perturbation of the gravitational potential (say $\delta_{_E}V_{\text{grav}}$)
\footnote{We reserve the notation $\phantom{f}\prime \leftrightarrow \derr{}$ and denote Lagrangian and Eulerian perturbations by $\delta_{_L}$ and $\delta_{_E},$ respectively.}.  In particular, our independent variables amount to
being coefficients derived from spherical-harmonic expansions of the form 
$$f(r,\theta,\phi)=\sum_{\ell=1}^{\infty}\sum_{m=-\ell}^{\ell}f_{\ell,m}(r)Y_{\ell}^m(\theta,\phi)$$ (cf. Section 5.7 \cite{svh}) for
spheroidal normal modes
along with Fourier expansions separating the time
variable:
Here, parameters $\gL_{\ell}$ $\df\ell(\ell+1)$ for given (harmonic) degree $\ell = 1, 2, \hdots,$  and (time) frequency of modes $\gs,$ which we take to be positive. The  remaining physical quantities, depending on $r,$ (also denoting the position function) are as follows:  
 $c$ is the adiabatic speed of sound; $g$ is the acceleration of gravity; and, $N$ is the Brunt-V\"ais\"ail\"a frequency, given by $N^2=-g(\frac{g}{c^2}+\frac{\rho'}{\rho}).$ Here, $N^2$ is taken to be non-negative to assure real $\gs$ and, in turn, dynamic stability;
 moreover, the case $N^2=0$ is called adiabatic (isentropic) equilibrium (cf. Sections 3.5, 13.2 \& 13.3 \cite{svh}).
Finally, we also use an alternative formulation via the change of variable $y_{_{LW}} \df$ $\gs^2 \eta$ $-$ $\Phi$  (Ledoux-Walraven - LW) with $\rho y_{_{LW}} $ $=$ $\delta_{_E}P,$ for pressure $P(r):$
We find this formulation convenient for high-frequency asymptotic estimates, $\gs^2$ $\rightarrow\infty.$
(Yet, our analysis of the resulting LW differential system will be limited since
particular solutions of ES redundantly arise after an application of integration by parts.)

If we  set $\vec{X}$ $=\left(u,\eta,\Phi,\Phi'\right)^{\top},$ the full system (\ref{otherfirst})-(\ref{otherthird}) can be 
written in the form
\begin{equation}\derr{}\vec{X}=A_{\gl,\gm}\vec{X}\label{full}\end{equation}
(introducing parameters $\gl,\gm$ used in Appendix A) for 
$$A_{\gl,\gm} \df\begin{bmatrix}\cA &\gl\vec{\cG}_{\gs} &\begin{bmatrix}0\\0\end{bmatrix}\\
\begin{bmatrix}0&0\end{bmatrix}&0&1\\
\gm\vec{\cC}^{\,\,\top}_\gs&\frac{\gk h+\gL_{\gz}}{r^{2}}&-\frac{2}{r} 
\end{bmatrix}
$$
where, setting $h\df r^2\rho/c^2$ 
 and  
$$\cA\df
\begin{bmatrix}\frac{g}{c^2}&\left[\gL_{\ell}-\gs^2\frac{r^2}{c^2}\right]\\ 
\frac{1}{r^2}\left(1-\frac{N^2}{\gs^2}\right)&\frac{N^2}{g}\ 
\end{bmatrix}
;\,\,\vec{\cG}_{\gs}\df\begin{bmatrix}\frac{r^2}{c^2}\\-\frac{N^2}{\gs^2g}
\end{bmatrix}
;\,\,\vec{\cC}_{\gs}\df\begin{bmatrix}\frac{\gk N^2\rho}{g}\\ \gs^2\gk h
\end{bmatrix}
$$ 
(For simplicity, we set $A\df A_{1,1}$.)
We express solution vectors $\vec{X}$ in terms of {\it component vectors} $\vec{Y}\df (u, \eta)^{\top}$  and $\vec{\Phi}$ $\df (\Phi,\Phi')^{\top}$  whereby $\vec{X}^{\top}$ $=(\vec{Y}^{\top},\vec{\Phi}^{\top})$
in block form. 
Here, the Cowling approximation is the ansatz that, under certain conditions, components $u,\eta$ of solutions to $\vec{X}'=A\vec{X}$ can be approximated by corresponding components of solutions to $\vec{Y}'=\cA\vec{Y},$ one such condition being that $\ell$ is large. As we
will see below, this anzatz amounts to neglecting $\Phi$ and, thus, decoupling equations within the system.

In order to describe our approach, further elaboration is needed:
 We reformulate the (full) system via equations of the form (dropping parameters  $\ell,\gs$)
\begin{align}\derr{}\vec{Y}&=\cA\vec{Y}+\Phi\vec{\cG}\label{formb1}
\\
\cL[\Phi]&\df  \derr{}\left(r^2\derr{\Phi}\right) -(\gL-\kappa h)\Phi=\vec{\cC}\cdot\vec{Y}\label{formb2}
\end{align}
which we combine to form an integro-differential equation 
\begin{equation}\cL[\Phi]=\cF[\Phi]+\vec{\cC}\cdot\vec{Y}_0\label{main}\end{equation}
so that, under appropriate conditions, the following hold:
\begin{equation}\cF[\Phi](r)=\int_a^bF(r,t)\Phi(t)\,dt \label{calF}\,\text{and}\end{equation} 
\begin{equation}F(r,t)=(\vec{\cC}(r))^{\top}\cM(r,t)\vec{\cG}(t)\label{F}\end{equation}
where $\cM$ is a Green's matrix for equation (\ref{formb1}); and,
the component vectors $\vec{Y}$ take the form 
\begin{equation}\label{Ys}
\vec{Y}=\int_a^b \cM(r,t)\vec{\cG}(t)\Phi(t)\,dt +\vec{Y}_0\df \vec{Y}_p+\vec{Y}_0 
\end{equation}
where $\vec{Y}_0$ is a solution to
\begin{equation}\label{residual}
\vec{Y}'=\cA\vec{Y}.
\end{equation}
Here, the associated (homogeneous) equation (\ref{residual}), what remains of (\ref{formb1}) when $\Phi\equiv 0,$  will be called the {\it residual equation} and its solutions $\vec{Y}_0$ will be called {\it residual solutions} whose components comprise the so-called Cowling approximation.
 
This ansatz, whereby $\Phi$ is negligible and $\vec{Y}_0$ closely approximates $\vec{Y},$ will be tested under 
large parameters, $\gL_{\ell},$ and/or $\gs^2.$ To this end, we establish some notation:
We introduce a (continuous) parameter 
$\gz>0$ in establishing some general asymptotic properties and estimates, sometimes replacing a discrete parameter. 
We use standard $O-$notation, on matrices as well as functions, where the
implied bounds are uniform on $[a,b],$ or on certain subintervals, and where the limit at infinity in the parameter is understood, unless otherwise specified. 
We may further choose to insert, drop or reassign parameters when held constant or otherwise redefined, with any alterations made clear, in our various cases. We will also introduce notation for special boundary conditions, as we develop below.

In our formulation of the approximation scheme we will invoke various multi-point boundary conditions, to be imposed on
components of solutions or together with those of residual solutions.
We will consider St\"{u}rm-Liouville operators (SL) in the form
$\fL[w]=$ $(p(r)w')'+q(r)w$ (cf.\cite{cl,birkrota}) with self-adjoint (SA) boundary conditions (BC) imposed on $w.$
For parameters $0\leq\gth_1<\pi,$ $0<\gth_2\leq \pi$ and $\vec{\gth}\df$ $(\gth_1,\gth_2)$  we denote
\begin{equation}\fS_{\fL}(\vec{\gth}):\,  \begin{bmatrix}-\cos\theta_1\\ p(a)\sin\theta_1\end{bmatrix} \cdot
\begin{bmatrix}w(a)\\w'(a)\end{bmatrix}
= \begin{bmatrix}-\cos\theta_2\\ p(b)\sin\theta_2\end{bmatrix} \cdot 
\begin{bmatrix}w(b)\\w'(b)\end{bmatrix}=0
\label{SLBC}
\end{equation}
We will also introduce  boundary conditions on component vectors $\vec{Y},\vec{Y}_0$ and associated notation:
Let $E_1,E_2$ denote standard $\Bbb R^2$ basis vectors and with $j,k=1,2$  denote
\begin{equation}\Pi_{j,k}: E_j \cdot \vec{Y}(a) =A;\, E_k\cdot \vec{Y}(b) = B
\label{Pi} \end{equation}
where such BCs (\ref{Pi}) will be said to be of $\Pi_{j,k}$ {\it type}.  Moreover, 
the associated vectors $\vec{Y}$ and $\vec{Y}_{0}$ of (\ref{Ys}) will be said to have {\em boundary agreement} $\Pi_{j,k}$ if 
(\ref{Pi}) holds with $A=B=0$ as applied to $\vec{Y}_p=$ $\vec{Y}-\vec{Y_{0}}.$ Moreover, we will apply the same
notation for such boundary conditions on component vectors $\vec{\Phi}.$ We note finally that any such boundary conditions or boundary agreement may depend on parameter $\gz$ by
varying $\theta_j$'s or by varying the interval $[a,b]$.

We will introduce a sufficient condition for which various two-point boundary conditions hold. 
Supposing $\vec{Y}_{0,1},$  $\vec{Y}_{0,2}$ form a basis of solutions to the residual equation,
we denote the associated determinant
\begin{equation}\notag
\cW_{j,k}(\gz)\df\left|\begin{matrix}u_{0,1}(r_j,\gz)&u_{0,2}(r_j,\gz)\\\eta_{0,1}(r_k,\gz)&\eta_{0,2}(r_k,\gz)
\end{matrix}
\right|
\end{equation} for $r_1=a,r_2=b$ (unless otherwise specified)
and are ready to introduce
\begin{prop}\label{cW} If $\cW_{j,k}\neq 0,$ then (\ref{Pi}) holds for any $A,B$ satisfied by unique $\vec{Y}$ and $\vec{Y}_0$ of (\ref{Ys}) so that 
$\vec{Y}$ and $\vec{Y}_0$ have $\Pi_{j,k}$ boundary agreement.
\end{prop}
\begin{proof}
The boundary conditions in each case are clearly admissible for the associated residual equation (\ref{residual}); and, standard arguments 
for two-point boundary problems (cf. \cite{col}) show that $\vec{Y}_p$ and $\vec{Y}_0$ as in (\ref{Ys}) are thereby uniquely determined.
\end{proof}
The following holds immediately from the construct (\ref{Ys}):
\begin{prop} Under the hypothesis of Proposition \ref{cW} any solution of $\Phi\in$ $C^2([a,b])$ of (\ref{main}) yields a unique
 solution to the system (\ref{formb1}),  (\ref{formb2}) via the correspondence $\vec{X}$ $\rightarrow(u,\eta,\Phi,\Phi')^{\top}$
for those boundary conditions of type $\Pi_{j,k}$ that are admissible to  (\ref{formb1}).
\end{prop}

We will let $\cC_{\gs}$ $\df$ $\text{diag}(\vec{\cC}^{\,\top}_{\gs})$ and, under appropriate conditions, we will find estimates of the form 
\begin{align}\|\Phi\|&\leq \fz\|\vec{\cC_{\gs}}\cdot\vec{Y}_0\|
\label{weakCA}\\
\|\cC_{\gs}(\vec{Y}-\vec{Y}_0)\|_{\infty}&\leq \cZ\|\cC_{\gs}\vec{Y}_0\|_{\infty}\label{strongCA}
\end{align}
under the standard $L^2$ and supremum norms, respectively, on $[a,b]$ - extended also to matrices: Here, the respective modulii $\fz,\cZ$ take positive values
and depend (perhaps) only on a large parameter.
\begin{dfn}
We will say that {\it weak-CA} ({\it strong- CA}) holds if, for a parameter $\gz>0$ (say), the modulii $\fz$ (resp. $\cZ$) tends to $0$ in the limit as $\gz$ $\rightarrow$ $\infty.$ 
\end{dfn}
Finally, we remark that the inequality $\|\vec{\cC}_{\gs}\cdot\vec{Y}_0\|\leq $ $\sqrt{b-a}\|\cC_{\gs} \vec{Y}_0\|_{\infty}$ 
suggests, considering right-hand sides, that
(\ref{strongCA}) is not a stronger condition than (\ref{weakCA}); yet, (\ref{strongCA}) conveniently leads directly to a supremum estimate (rather than $L^2$) that directly compares  $\eta-\eta_0$ to 
 $\eta_0$ and, if $N^2>0,$ compares $u-u_0$ to $u_0.$ Moreover, we find these formulations to be convenient in explicit calculations. 

By imposing appropriate two-point boundary conditions, our main approach in this work will be to restrict the problem to that where $\cL$ is self-adjoint and where $\cF$ is, or is closely approximated by, a symmetric operator.  We will then be using a theorem from \cite{a} which applies to integro-differential equations of the general form
 \begin{equation}\fL[y]=\fF[y] +f_0\label{basiceqn}\end{equation}  
 for equations involving linear operators on a Hilbert space $\cH$ (in our case $L^2([a,b])$  and on subspaces thereof: 
 \begin{thm} \label{thm1}

Suppose $\fD$ is a closed subspace of a Hilbert space $\BH$ with $\fL:$ $\fD\rightarrow \BH$ and $\fF:$ $\BH\rightarrow$ $\BH$ as in (\ref{basiceqn})
where $\fL$ is self-adjoint.
\begin{itemize}
\item[i)]
Suppose further that $\fF$ is symmetric and that there is a constant $\gamma>0$ so that
 \begin{equation}
\left<\fL[y]-\fF[y],y \right><-\gamma \|y\|^2\notag
\end{equation}
 $\forall y\in \fD.$ Then,  (\ref{basiceqn}) has a unique solution $y_*\in\fD:$ Moreover, the solution satisfies 
 \begin{equation}\|y_*\|\leq \frac{1}{\gamma}\|f_0\|\label{mainresult}\end{equation} 
\item[ii)]
 Suppose instead that there is a constant $\gep>0$ and a symmetric operator $\fF_s$ so that 
$\|\fF-\fF_s\|\leq\gep$ and that 
\begin{equation}
\left<\fL[y]-\fF_s[y]+\gep y ,y \right><-\gamma \|y\|^2\label{mainineq2}
\end{equation}
$\forall y\in \fD.$ Then, (\ref{mainresult}) holds uniquely for some $y_*$ as in i).
\end{itemize}
 \end{thm} This is a  slight modificationof Theorem 2.6 of \cite{a} (see also section 5.2 of \cite{c}) to accommodate our applications.
 \begin{proof}
 Here, the Gateaux derivative of the operator $y\rightarrow$ $\fF[y]+f_0$ yields the operator $\fF[y]$ to which
we may equate the operators $B^-$ and $B^+.$ We may further set the spaces $X^+=\{0\}$ and $X^-=\BH,$ whereby
a clear modifiction of Lemma (2.5)  \cite{a} for $Q^+=0$ leads to the conclusion of case i). The result of case ii) likewise follows by setting
 $B^{\pm}=$ $\fF_s\pm\gep I,$ respectively.
 \end{proof}
We note that Theorem \ref{thm1}, the main inspiration of our approach,
may not be the only means by which one establishes the existence and uniqueness of solution components $\Phi$ and that estimates
(\ref{mainineq2}) 
we derive may be obtained by standard eigenvalue and norm estimates. For instance, as we show in Section \ref{s6} (see Appendix B), more specific boundary 
conditions on the components $\vec{\Phi}$ or on both $\vec{Y}$ and $\vec{\Phi}$  may be ascertained via large-parameter asymptotics.

Our main goals will be to verify the following properties
which we describe by 
the following:
\begin{dfn}  A solution $\vec{X}(\cdot,\gz)$ to (\ref{otherfirst})-(\ref{otherthird}) is $\cC$-approximable with respect to $\gz$ (under given BC) if for the corresponding $\vec{Y},\Phi,\vec{Y}_0$
both weak- and strong- CA hold on interval(s) $\cI.$
Also, a basis  $\{\vec{X}_i\}_{i=1}^4$ of said system will be called $\cC$-approximable if each 
basis vector $\vec{X}_i$ is $\cC$-approximable.
\end{dfn}
The article is organized as follows: In Section \ref{s2} we treat the case $N^2= 0$ on $[a,b],$ including various subcases of high degree and/or high frequency, each of which apply to
methods used in the rest the article.
In Section \ref{s3} we then treat the cases where $N^2> 0$ which involve the appearance of terms of various order in the 
parameter and various symmetric operators in the associated integral equation. In Section \ref{s4}, we compute example asymptotic
estimates, to leading order in the parameter, as checks against necessary conditions of some main results.
In Sections \ref{s5} and \ref{s6} (Appendices A and B), we establish 
asymptotic estimates which apply throughout the article, involving those of both full-system and the residual equations.
 \section{Adiabatic Equilibrium}\label{s2}
We will first consider the special case where $N(r)$ identically vanishes on $[a,b],$ whereby $g/c^2 =-\rho'/\rho.$  We may substitute accordingly and apply an integrative factor as we combine (\ref{otherfirst}) and (\ref{othersecond}) to arrive at the following system, with $u=r^2\eta':$ 
 \begin{align}
\cJ[\eta]&\df\derr{\hspace{0in}}\left(r^2\rho\derr{\eta}\right) - \left[\gL_{\ell}-\gs^2\frac{r^2}{c^2}\right]\rho\eta =\frac{r^2\rho}{c^2}\Phi
\label{rcEC1}\\
\cL[\Phi]&\df\derr{\phantom{}}\left(r^2\derr{\Phi}\right)-\left[\gL_{\ell}-\frac{\gk r^2\rho}{c^2}\right]\Phi =\gs^2\frac{\gk r^2\rho}{c^2}\eta\label{rcEC2}
\end{align}
We note that both operators $\cJ$ and $\cL$ are of SL form where self-adjoint boundary conditions will be expressed via (\ref{SLBC}) .
We set $\fJ$ to be the associated Green's function for the operator $\cJ$ so that $\eta$ satisfies $\eta=\eta_p+\eta_0$ where
$\eta_p=\fJ[h\Phi]$ and
\begin{equation}\label{adiabaticinteqn}\cL[\Phi]=\gs^2\kappa h(\fJ[h\Phi]+\eta_0)\df\cF[\Phi]+\gs^2\kappa h\eta_0,
\end{equation}
whereby the components of $\vec{Y}_0$ have the simple relation 
$u_0$ $=\derr{(r^2\eta_0)}$ for $\eta_0$ $\in$ ker$\cJ.$
 We will study this system under various parametrizations which will be indicated in the notation accordingly.  
\subsection{High  Degree and Fixed Frequency}
We insert the parametrization $\gL_{\ell}$ $\rightarrow$ $\gz^2$ for large $\gz>0$ in  (\ref{rcEC1}) and (\ref{rcEC2}) on which both the operators $\cJ_{\gz}$ and  $\cL_{\gz}$
now depend. As an immediate result, we establish a weak-CA estimate involving two-point boundary conditions in a slightly general setting in the following
\begin{prop} \label{frmrthm2}
Given $\vec{\theta},\vec{\vartheta}$ (fixed) and any $\eta_0$ $\in$ ker$\cJ_{\gz}$ 
 there is a solution $\Phi$ to (\ref{adiabaticinteqn}) whereby 
\begin{itemize}
\item[i)]  $\Phi,\eta_p$ satisfy $\fS_{\cL}(\vec{\theta}),$ $\fS_{\cJ}(\vec{\vartheta}),$ respectively; and, 
\item[ii)]
$\|\Phi(\cdot,\gz)\|\leq\fz({\gz})\| h\eta_0(\cdot,\gz)\|
$
for $h=r^2\rho/c^2$ where \newline
$\fz({\gz})$ $= \gz^{-2}$ $(1+o(1))$ as $\gz\rightarrow \infty.$
\end{itemize}
\end{prop} 
\begin{proof} From (\ref{rcEC1}) and (\ref{rcEC2}) we find that of $\cF_{\gz}$ takes the form 
$$\cF_{\gz}[y]=\gs^2\kappa h \tilde{\fJ}_{\gz}[hy]$$
where $\tilde{\fJ}_{\gz}$ is the symmetric Green's operator
associated with $\cJ_{\gz}$ under $\fS_{\cJ}(\vec{\theta})$. 
We may set $\gm_0$ to be the greatest eigenvalue (depending on $\vec{\theta}$) of the operator $\cJ_{0}$ 
with weight $\rho$ whereby the following holds for all $y\in \cC^2([a,b])$ satisfying $\fS_{\cJ}(\vec{\vartheta}):$
\begin{align}\notag
\langle \cJ_{\gz}[y],y\rangle& \leq \langle \gm_0 \rho y,y\rangle - \langle \gz^2 \rho y,y\rangle\\\notag
&\leq (-\gz^2 +|\gm_0|)\rho_{-} \|y\|^2
\end{align}
for $\rho_{-}\df\min_{r\in[a,b]}\rho(r).$
Then, given $\gz> \sqrt{|\gm_0|},$ the operator $\tilde{\fJ}_{\gz}$ 
extends to a bounded linear operator on $L^2([a,b])$ and
that $\cF_{\gz}$ extends to a symmetric operator whereby
\begin{equation}\|\cF_{\gz}\|\leq \frac{\gs^2\kappa \|h\|^2/\rho_{-}}{\gz^2 - |\gm_0|}\label{cFEst}\end{equation}
Letting $\gn_0$ be the greatest eigenvalue of $\cL_{\gz}$ associated with $\fS_{\cL}(\vec{\theta}),$ 
result ii) is now clear from Theorem \ref{thm1} with choice of $\gamma$ given by
\begin{equation}\gamma(\gz) = \gz^2 - |\gn_0|- \|\cF_{\gz}\|\label{gammaest}\end{equation}
which is positive for all  sufficiently large $\gz.$ Thus, result  ii) thereby holds for $\fz({\gz})$ $=1/\gamma(\gz).$
Finally, i) holds for $\eta_p=\eta-\eta_0=$ $\cJ_{\gz}[h\Phi].$
\end{proof}

We now address
the integral operator of (\ref{Ys}), involving here the construction of a Green's matrix of (\ref{formb2}) to develop associated strong-CA estimates:
We state (see Appendix A) that
equation (\ref{formb1}) has fundamental matrices $\fN_{j,k}:$ $j,k =1, 2$ which satisfy
\begin{equation}\fN_{j,k}=\cP\left(I+O(\gz^{-1})\right)\fM_{j,k}\label{fN}
\end{equation}
for $\fM_{j,k}\df\left[\vec{\fM}_{j,1},\vec{\fM}_{k,2}\right]$ and
\begin{equation}\notag\vec{\fM}_{1,k}(r,\gz) \df \begin{bmatrix}\sinh(\zeta \Theta_{r_k}(r))\\ \cosh(\zeta \Theta_{r_k}(r))\end{bmatrix},\,
 \vec{\fM}_{2,k}(r,\gz) \df \begin{bmatrix}\cosh(\zeta \Theta_{r_k}(r))\\ \sinh(\zeta \Theta_{r_k}(r))\end{bmatrix};
\end{equation} \begin{equation}\Theta_{r_j}(r)\df \int_{r_j}^r\frac{dt}{t};\,  \cP = \text{diag}\left(\sqrt{r/\rho};\,\frac{1}{\gz\sqrt{r\rho}}\right)\label{cP}
\end{equation}
(with $\Theta$ to denote a general antiderivative). 
 
We further develop our estimates as follows: The $\fM_{i,j}$ satisfy
 \begin{equation}\text{adj}(\fM) =\frak{I} \fM\frak{I}^{\top}\label{symmprop}\end{equation}
 for $\fI\df[E_2,-E_1]$  whose determinants 
 $\Delta_{i,j}(\gz)$ $\df$ $\det(\fM_{i,j})$ are each independent of $r,t$ and satisfy $|\Delta(\gz)|$ $\geq$ $(e^{\zeta\Theta_a(b)}-1)/2.$
Furthermore, denoting by $\chi(r,t)$ the characteristic function of the subinterval $\{r|r\leq t\}$ of $[a,b]$ and $\fX(r,t)$ $\df$ $\text{diag}\left(-\chi(r,t),\,\chi(t,r)\right),$
the associated kernel as in (\ref{Ys}) is given by $\cM_{\gz}(r,t)=$
\begin{equation}\label{Mformula}
 \cP(r)\fM(r,\gz)\fX(r,t)\text{adj}(\fM)(t,\gz)\cP^{-1}(t)\vec{\cG}(t)(I+O(\gz^{-1})),
\end{equation}
to yield $F_{\gz}$ as in (\ref{F}).  
We now introduce the following
\begin{prop}\label{prop1}
 The associated $\vec{Y}_p$ of Proposition \ref{frmrthm2} may be chosen to satisfy $\Pi_{j,k}$ in which case, for $f$ $=\frac{h}{\sqrt{r\rho}},$
the associated kernel $F_{\gz}$ of operator
 $\cF_{\gz}$ of (\ref{calF} ) satisfies
$$F_{\gz}(r,t)=\frac{\kappa \gs^2}{\gz\Delta(\zeta)} f(r)f(t)J^{(s)}_{\gz}(r,t) \left(1+O(\zeta^{-1})\right)$$
 where $J^{(s)}_{\gz}(r,t)$ is a symmetric function and $\frac{J^{(s)}_{\gz}}{\Delta}$
is uniformly bounded for $(r,t)\in$ $[a,b]\times[a,b]$ for all sufficiently large $\gz$.\end{prop}
\begin{proof}
From (\ref{F}) we compute, noting that $\fX(r,t)\fI$ $=\fI^{\top}\fX(t,r),$
\begin{align}
J^{(s)}_{\gz}(r,t)&=E_2^{\top}\fM(r,\zeta) \fX(r,t)\text{adj}(\fM)(t,\zeta) E_1
\label{Jtil1}\\
 &=E_2^{\top}\fM(r,\zeta) \fX(r,t)\fI\left(\fM(t,\zeta)\right)^{\top}\fI^{\top} E_1\notag\\
&=
\left(E_1^{\top}\fI\fM(r,\zeta)\right) \fI^{\top}\fX(t,r)\left(\fM(t,\zeta)\right)^{\top}E_2\notag\\
&=
E_1^{\top}\text{adj}(\fM)(r,\zeta)\fX(t,r)\left(\fM(t,\zeta)\right)^{\top}E_2\notag\\
&=\left(J^{(s)}_{\gz}(t,r)\right)^{\top}=J^{(s)}_{\gz}(t,r)\notag
\end{align}

Since  $\exp(\zeta\Theta_a(r))\chi(r,t)\exp(\zeta\Theta_b(t))$ $\leq$ $\exp(\zeta\Theta_b(a))$ on the square $(r,t)$ $\in$ $[a,b]\times [a,b],$ it follows from (\ref{Jtil1}) that
$J^{(s)}(r,t)$ $=O(\Delta(\gz))$ holds on this domain; and, hence, the proof is complete. 
\end{proof}

We summarize the main results along with some specific estimates involving components of $\vec{Y},\vec{Y}_0$ in  
\begin{thm}\label{thm3} In the case $N^2=0$ on $[a,b]$ the following hold for the system (\ref{otherfirst})-(\ref{otherthird})
on the interval
for any fixed $\vec{\theta},j,k$ 
for all sufficiently large $\ell$ $\rightarrow \infty:$
\begin{itemize}
\item[i)] Solutions $\vec{X}$ may be chosen so that $\Phi$ satisfies $\fS_{\cL}(\vec{\theta})$
and $\vec{Y}$ has BCs of type $\Pi_{j,k}$ whereby,
\item[ii)] $\vec{X}$ is $\cC$ approximable with $\vec{Y}.\vec{Y}_0$ having
 $\Pi_{j,k}$ boundary agreement;
\item[iii)] the system has a $\cC$-approximable basis; and,
\item[iv)] 
 the implied strong-CA estimates of ii) in particular satisfy
\begin{equation}
\|\eta-\eta_0\|_{\infty}\leq \frac{C_1}{\gz^3}\|\eta_0\|_{\infty},\,\|u-u_0\|_{\infty}\leq \frac{C_2}{\gz^2}\|\eta_0\|_{\infty}
\label{difffest}\end{equation}
for some positive constants $C_1,C_2.$  
\end{itemize} 
\end{thm} 
\begin{proof}
The weak CA under BC of i) follows from i) of Proposition \ref{frmrthm2} and from Propositon \ref{prop1}.
For each $j,k$ the corresponding $\vec{Y}_p$ satisfy (\ref{Ys}) but for
$\cM(\cdot,{\gz})$ $=$ $\fN_{j,k}\fX\fN_{j,k}^{-1}$
and, hence, $\Pi_{j,k}$  (\ref{Pi}).
Furthermore, 
$$\text{adj} (\fN(\cdot,{\zeta})) \vec{\cG}=(v_1, v_2)^{\top}$$ for functions $v_1$ and $v_2$ which are majorized by $\chi(r,t) \exp(\zeta \Theta_b(t))$ and $\chi(t,r)\exp(\zeta \Theta_a(t)),$ respectively.  
So, each element of $M(\cdot,{\gz})\vec{\cG}$ is majorized by 
\begin{equation}\label{kerest} \exp(\zeta(\Theta_a(r)+\Theta_b(t)))\chi(r,t) +\exp(\zeta(\Theta_b(r)+\Theta_a(t)))\chi(t,r)
\end{equation}
and hence by $\Delta(\zeta).$ Thus, the supremum over the square satisfies $\|\cM(\cdot,{\gz})\cG\|_{\infty}$ $=O(1)$ as $\gz\rightarrow \infty$ so that
\begin{equation}\label{Ypest}
\|\vec{Y}_p\|_{\infty}\leq \sqrt{b-a}\|\cM(\cdot,{\gz})\cG\|_{\infty}\|\Phi\|
\end{equation}
and that the estimates (\ref{difffest}) follow as in Proposition \ref{prop1}, confirming ii) and iv).
Finally,
iii) follows from Propositon \ref{lastprop1}, with the fundamental solution $X\fW$ whose column vectors correspond to various (fixed) $\fS_{\cL}$ and, thus, by i).
\end{proof}

\subsection{High Frequency and Fixed Degree}
We follow procedures similar to those above but with the presence of the parameter $\gs$ in the operator of (\ref{rcEC1}), now denoted by $\cJ_{\gs}$
(dropping $\gz$),
whereby the estimates 
are now more complicated: Indeed, with $0$ as a possible eigenvalue, invertibility and Green's-function estimates cannot be so readily assured;
moreover, in light of Corollary \ref{CAcor} (see Appendix A), we cannot expect both weak and strong estimates to hold over a fixed interval 
$[a,b].$
 However, we find that we may establish estimates in similar manner by adjusting boundary agreement $\Pi,$ depending on $\gs,$
and by allowing the size of the interval on which estimates hold to also vary accordingly: We therefore introduce the notation
$\fS(\gs),\Pi(\gs)$ to indicate the dependence of the respective boundary conditions on the parameter $\gs$ and the
notation $\cI_{\gb}$ $\df$ $[a,\gb]$ for $a<\gb\leq b$ to denote subintervals of $\cI\df[a,b]$ with like subscript 
to indicate related quantities their dependence on $\gb$.
We begin with
\begin{prop}\label{innerprodest} Boundary conditions $\fS_{\cJ}(\gs)$ as in (\ref{SLBC}),  depending on $\gs,$ can be chosen so that
\begin{equation}\notag\left|\left<\cJ_{\gs}[y],y\right>\right| \geq \fc \gs \left<hy,y\right>\end{equation}
holds for $y\in \cC^2([a,b])$ satisfying $\fS_{\cJ}(\gs)$ for
some  constant $\fc>0$ independent of $\gs.$
\end{prop}
\begin{proof}
Consider the eigenvalue problem
$$\cJ_0[y]= -\gl h y$$ with BCs of the form $\fS_{\cJ}$ as follows:
If we fix $\theta_1,\theta_2$ so that $\sin(\theta_1)\cdot\sin(\theta_2)$ $\neq 0,$ 
then, as $n\rightarrow \infty$ the eigenvalues $\gn_n$ asymptotically satisfy 
\begin{equation}\sqrt{\gn_n}=\frac{(n-1/2)\pi}{L} +O(n^{-1})
\label{gns}
\end{equation} 
where $L=\int_a^b 1/c(t)dt.$ 
For $\theta_2$ as above and $\theta_1$ $=0$ (say), the eigenvalues $\gm_n$ satisfy
\begin{equation}\sqrt{\gm_n}=\frac{n\pi}{L} +O(n^{-1})
\label{gms}
\end{equation} It follows that 
$\exists n_0\in {\Bbb N}^+$ such that
$$|\sqrt{\gn_k}-\sqrt{\gm_j}|> \frac{\pi}{3L}
$$
and $\gn_k,\gm_j>(\frac{\pi}{2L})^2$ for all $j,k\geq n_0.$ These results likewise hold for $\theta_1,\theta_2$ chosen vise-versa. 
In the case $\sin(\theta_1)$ $=\sin(\theta_2)$ $=0$ the corresponding eigenvalues satisfy (\ref{gns}).

So, we may take $\fS_{\cJ}(\gs)$ to be $\fS_{\cJ}(\theta_1,\theta_2),$ as in (\ref{gns}), if $|\gs -\sqrt{\gm_k}|$ $<\frac{\pi}{4L}$
for some $k$ and $\fS_{\cJ}(\gs)$ to be $\fS_{\cJ}(0,\theta_2),$ as in (\ref{gms}), otherwise; and, let $\gl_k(\gs)$ be the corresponding
eigenvalue $\gm_k$ or $\gn_k$ now depending on $\gs.$
Accordingly, the eigenvalues $\gl_k(\gs)$ satisfy
\begin{equation}\notag
|\gs^2-\gl_k(\gs)|\geq \frac{\pi}{12L} (\gs+\sqrt{\gl_k(\gs)})>\frac{\gs\pi}{12L} 
\end{equation} holds $\forall k$ for all sufficiently large $\gs.$ Thus, there is a constant $C>0$ independent of sufficiently large $\gs$ whereby 
\begin{align}
\left|\left<\cJ_{\gs}[y],y\right>\right|& \geq\left|\left<\cJ_0[y],y\right>+\left<\gs^2hy,y\right>\right|-\left|\left<\gL yc^2h/r^2,y\right>\right|\label{propeqn1}\\
&\geq (\frac{\gs\pi}{12L}-C)\left<hy,y\right>\label{propeqn2}
\end{align} so that the desired result holds.
\end{proof}
We note, as it will be important later, that the constant $C$ of (\ref{propeqn2}) can be chosen to be independent of $\gb:$ $a<\gb\leq b$ on intervals $\cI_{\gb}.$ 

We proceed further to develop more detailed constructions.  From 
(\ref{gms}), (\ref{gns}) we may chose $n_0$ is so large that $\forall n\geq n_0$
$$|\sqrt{\gn_n}-\frac{(n-1/2)\pi}{L}|, |\sqrt{\gm_n}-\frac{n\pi}{L}|<\frac{\pi}{12L} 
$$
and introduce fundamental matrices $\fN_{j,k}$ as in (\ref{fN}) but with
\begin{equation}\label{nextcP}
\Theta_{r_j}(r)\df\int_{r_j}^r\frac{dt}{c(t)},\,\cP=\left(r/\sqrt{\rho c},\,\sqrt{c}/(\gs r\sqrt{\rho })\right),
\end{equation}
 (so that $L=\Theta_a(b)$), and $\fM_{j,k}\df$ $\left[\vec{\fM} _{j,1},\vec{\fM}_{k,2}\right]$  
where
\begin{equation}\label{os}\vec{\fM}_{1,k}(\cdot,\gs) \df \begin{bmatrix}\sin(\gs \Theta_{r_k})\\ \cos(\gs \Theta_{r_k})\end{bmatrix} ,\,
 \vec{\fM}_{2,k}(\cdot,\gs) \df \begin{bmatrix}-\cos(\gs \Theta_{r_k})\\ \sin(\gs \Theta_{r_k})\end{bmatrix}
\end{equation}
(See Propositon \ref{LWest}, Appendix A, with $\gl=0$.)
In turn, the $\fM_{j,k}$ satisfy (\ref{symmprop}) and the respective determinants satisfy
\begin{equation}\label{dets}
|\Delta_{j,k}(\gs)|= \begin{cases}|\cos(\gs L) |&: j\neq k \\|\sin(\gs L)|&: j=k
\end{cases}
\end{equation}
Now, indices $j_m, k_m:$ $m=1,2,\hdots$ can be chosen so that $|\Delta_{j_m,k_m}(\gs)|$ $\geq \frac{1}{2}$ when
for $\gs\in I_m:m\geq n_0$ where
 \begin{equation}\notag
I_{2m}\df(\sqrt{\gm_{m}}-\frac{\pi}{3L},\sqrt{\gm_{m}}+\frac{\pi}{3L});\, I_{2m-1}\df (\sqrt{\gn_{m}}-\frac{\pi}{3L},\sqrt{\gn_{m}}+\frac{\pi}{3L})
\end{equation}
 So, we can form (dropping indices) a fundamental matrix $\fN(\cdot,\gs),$ piecewise defined with respect to $\gs,$ with columns $\vec{\fN}_1,\vec{\fN}_2$ satisfying conditions  $\Pi(\gs)$ in accord with (\ref{Pi}) and whose determinant satisfies $|\Delta(\gs)|$ $\geq\frac{1}{2}$ for all sufficiently large $\gs.$ 
In turn, we find via (\ref{Mformula}) that the following norms over the square satisfy (as $\gs\rightarrow\infty$)
\begin{align}
\notag\|\fN(r,\gs)\cX(r,t)\fN^{-1}(t,\gs)\vec{\cG}_{\gs}(t)\|_{\infty}&=O(\gs^{-1})\\
\notag\|\cC_{\gs}(r)\fN(r,\gs)\cX(r,t)\fN^{-1}(t,\gs)\vec{\cG}_{\gs}(t)\|_{\infty}&=O(\gs)
\end{align}

From the above results, we conclude
\begin{prop}\label{prop2}
 The associated $\vec{Y}_p(\cdot,\gs)$ of (\ref{Ys}) may be chosen to satisfy boundary condition of the form $\Pi(\gs)$
 whereby the associated kernel $F$ of (\ref{calF}) for $f$ $=\frac{h}{r}\sqrt{\frac{c}{\rho}}$ satisfies
\begin{equation}\notag
F_{\gs}(r,t)=\frac{\gs\kappa}{\Delta(\gs)} f(r)f(t)J^{(s)}_{\gs}(r,t) \left(1+O(\gs^{-1})\right)\,\, ({\text as}\,  \gs\rightarrow \infty)\end{equation}
where $J^{(s)}_{\gs}(r,t)$ is symmetric and $J^{(s)}_{\gs}/\Delta$ is uniformly bounded for $(r,t)\in$ $\cI\times\cI$ and  $\gs$ sufficiently large. 
\end{prop}
Now, concerning the corresponding residual solutions, we state
\begin{prop}
Given sufficiently large $\gs>0,$ $\vec{Y}_0$ may be uniquely chosen to have boundary agreement $\Pi(\gs)$ with $\vec{Y}_p$ of Proposition \ref{prop2}.
\end{prop}
\begin{proof} From Proposition \ref{LWest} (see Appendix A) in the case $\gl=0$ one finds that the result follows from Proposition \ref{cW}.
\end{proof}

We now impose boundary conditions on the operator $\cL$ which will also vary with $\gs$ in order to assure large eigenvalues.  It will suffice to study a simpler operator: 
Let the differential operator $\tilde{\cL}$ be given by $\tilde{\cL}[y]\df (r^2y')'$
and let $\Pi_{j,k;\gb}$ denote boundary conditions $\fS_{\cL}$ that are homogeneous
on endpoints of intervals $\cI_{\gb}$ and assigned in a manner similar to those of Proposition \ref{innerprodest} as we state
\begin{lem}\label{FSLem}
Corresponding to boundary conditions of the form $\Pi_{j,k;\gb}:$ $(j, k)\neq (2,2),$ subspaces
 $\cD_{\gb}$ of $\cC^2(\cI_{\gb})$ exist such that  the following hold for any $a<\gb<b:$
\begin{itemize}
\item[i)]  $\cD_{\gb}$ is a dense subspace of $\cC^2(\cI_{\gb})$ on which $\tilde{\cL}$ is SA
\item[ii)] $\langle \tilde{\cL}[y],y\rangle\leq -\gt_{\gb} \|y\|$ $\forall y\in \cD_{\gb}$ where
\item[iii)] $\gt_{\gb}$ is strictly increasing s.t. $\gt_{\gb}\rightarrow \infty$ as $\gb\rightarrow a^+.$
\end{itemize}
\end{lem}
\begin{proof}
Consider as special case the BVs $y(a)=y(\gb)=0$ and the space $\cD_{\gb}$ given by the span of $\{\phi_{\gb,j}\}_{j=1}^{\infty}$ for
$
\phi_{\gb,j}(r)=\sin(\sqrt{\gl_{\gb,j}}\Upsilon(r))
$ 
with $\Upsilon(r)$ $\df\int_a^r \frac{dt}{t^2}$
and  $\sqrt{\gl_{\gb,j}}$ $\df$ $j\pi/\Upsilon(\gb).$ By the density of trigonometric polynomials and diffeomorphic properties of $\Upsilon,$
property i) follows.
Here, we find that the $\phi_{\gb,j}$ satisfy
\begin{equation}\notag\int_{\cI_{\gb}}\tilde{\cL}[\phi_{\gb,j}](r)\phi_{\gb,j}(r)dr =-\gl_{\gb,j}\int_{\cI_{\gb}}\phi_{\gb,j}^2(r)\Upsilon'(r)dr
\end{equation}
and that ii) follows.
Statement iii) also follows whereby
$1/\Upsilon(\gb)> \frac{a^2}{\gb-a}$ holds.
The arguments for all such conditions $\Pi$ are essentially the same, mutatis-mutandis, for homogeneous Dirichlet BCs in at least one endpoint. 
\end{proof} 

Following Lemma \ref{FSLem}, we will set $\gb$ to depend on $\gs$ in 
\begin{equation}\label{beta}\gb_{\gs}\df a+z\gs^{-\ga}
\end{equation}
for some positive constants $z,\ga.$ Moreover, we will simply denote by $\Pi_{\gb}(\gs)$ those boundary conditions/agreements of the form $\Pi(\gb)$ but applied to endpoints of $\cI_{\gb};$ and, we likewise denote by
by $\Pi_{j,k;\gb}$ those of the form $\Pi_{j,k},$ but imposed on $\vec{\Phi}$ at the endpoints
of $\cI_{\gb}$ - those of spaces $\cD_{\gb}$ as above.

\begin{thm}\label{limitwca}
Under the adiabatic condition, the system (\ref{otherfirst})-(\ref{otherthird}) satisfies the following for sufficiently large
$\gs\rightarrow \infty$ and small $z>0$ for some residual solutions $\vec{Y}_0$ and for $\gb_{\gs}$ as in (\ref{beta});
\begin{itemize}
\item[i)] (Three-point BC) solutions $\vec{X}$ are $\cC$-approximable for those whose components $\vec{\Phi}$ satisfy $\Pi_{j,k;\gb_{\gs}}$ ($j,k$ fixed)
where $\vec{Y},\vec{Y}_0$ have boundary agreement of the form $\Pi(\gs);$
\item[ii)] (Two-point BC) solutions $\vec{X}$ are $\cC$-approximable for $\vec{\Phi}$ as in i) but  
where $\vec{Y},\vec{Y}_0$ have boundary agreement of type $\Pi_{\gb_{\gs}}(\gs);$ 
\item[iii)] the system has a $\cC$-approximable basis in either sense i), ii); 
\item[iv)] the implied strong-CA estimates in cases i), ii) satisfy
\begin{equation}\label{diffest}
\|\eta-\eta_0\|_{\infty;\gb_{\gs}}\leq \frac{C_1}{\gs^{3\ga}}\|\eta_0\|_{\infty;\gb_{\gs}},\,\|u-u_0\|_{\infty;\gb_{\gs}}\leq \frac{C_2}{\gs^{3\ga+2}}\|\eta_0\|_{\infty;\gb_{\gs}},
\end{equation} for $\ga=1/2, 1,$ respectively, for  
some positive constants $C_1,C_2,$ whereby 
\item[v)] the following hold uniformly for $\gs$ on compact subsets of $\BR^+$
$$\|\eta-\eta_0\|_{\infty;\gb_{\gs}}, \|u-u_0\|_{\infty;\gb_{\gs}}=O(z)\,({\text as}\, z\rightarrow 0^+)
$$
\end{itemize}
\end{thm}
\begin{proof}
We will show that in cases i) and ii) the moduli may take the form
\begin{equation}\label{themods}
\fz(\gs) = \frac{K_1z}{\gs^{2\ga}};\cZ(\gs)=\frac{K_2z}{\gs^{3\ga}}
\end{equation}
to hold for all sufficiently large $\gs$ and sufficiently small $z>0$ over 
various $\cI_{\gb},$ starting with $\fz:$
In each of these cases it follows from 
Lemma \ref{FSLem} there is a positive constant $K_1$ so that
$\gt_{\gb_{\gs}}> K_1\frac{\gs^{2\ga}}{z^2}$ when $a<\gb\leq b.$
In case i) from Proposition \ref{innerprodest} it follows that, for appropriate BC, there is a positive constant $K_2$ so that 
$$|\langle\cF_{\gs}[y],y\rangle_{\gb}|\leq K_2\gs\|y\|^2_{\gb}$$ holds for $y\in \cD_{\gb},$  
in which case the form of $\fz$ applies for $\ga=1/2$ for all sufficiently large $\gs$
uniformly for some sufficiently small $z.$ 
In case ii) we note that there is a positive constant $C$ so that 
$L>$ $C(a-\gb)$ for $0<\gb\leq b$ and that we may modify 
Proposition \ref{innerprodest} via
(\ref{propeqn2}) 
to show that there is a positive constant $K_3$ so that 
$$|\langle\cF_{\gs}[y],y\rangle_{\gb}|\leq z^{-1}K_3\gs^{\ga+1}\|y\|^2_{\gb}$$  holds for $y\in \cD_{\gb_{\gs} }$  
for all sufficiently small z and large $\gs,$ in which case the form $\fz$ applies for $\ga=1$ for all sufficiently large $\gs.$
Indeed, as seen from (\ref{dets}), $z$ may be further chosen to be so small that $\Pi$ need not even need to vary with $\gs.$

We now proceed with the estimates of $\cZ:$
For case i) we apply Proposition \ref{prop2} 
whereby the strong-CA estimates follow.
For case ii)
we may follow Proposition \ref{prop2} whereby it is clear that
$\|J_{\gs}\|_{\infty}$ is bounded independent of $r_2=\gb$ and $\gs$ so that boundary conditions in case ii)
imposed on $\vec{Y}_p$ yield, as in (\ref{Ypest}),
 $$\|\cC_{\gs}\vec{Y}_p\|_{\infty;\gb}\leq K_4\sqrt{\gb-a}\|\Phi(\cdot,\gs)\|_{\gb}\leq (\gb-a)\fz(\gs)\|\cC_{\gs}\cdot\vec{Y}_0\|_{\infty;\gb}
$$ uniformly for $a<\gb\leq b,$
whereby the results follow.

Finally, result iii) follows from Propositions \ref{lrgsig} and \ref{lastprop} (see Appendix B) and the fact that each column of
$X\fW^{-1}$ or of $X\fV^{-1}$ belongs to a subspace $\cD_{\gb}$ as in Lemma \ref{FSLem}.
The results iv) and v) follow from (\ref{themods}) in each of the cases i) and ii).
 \end{proof}
Several remarks are in order:
The spaces $\cD_{\gb}$ of Lemma \ref{FSLem}, if extended to the interval $\cI,$ can be considered to be projections of
components $\Phi$ onto subspaces of functions involving large numbers of zeros on $[a,b]$ in the sense of the St\"urm Oscillation Theorem.
Such a condition, termed as high
{\it radial order} (see cited works following (\ref{otherthird}) in Section \ref{introsec}),
is also one whereby the Cowling approximation is expected to be accurate.
Although tangentially related, such investigation is beyond the scope of the present work.

\subsection{High Degree and High Frequency} 
We now let both $\gL_{\ell}$ and $\gs$  depend on a single parameter $\gz$ as $\gz\rightarrow$ $\infty$ to address results analogous to those of Theorems \ref{thm3} and \ref{limitwca} in the following cases:
If we replace the parameters $\gL_{\ell},\gs^2$ by $\gz^2,  z\gz,$ respectively, then by procedures similar to those of (\ref{Beqn})  (see Appendix A)
applied to (\ref{residual}), we find estimates for matrices $\fN$ as in (\ref{fN}) but with
$\Theta_{r_k}$  replaced by
 $$\Theta_{r_k}(r,\gz)\df\int_{r_k}^r\frac{1}{t} -\frac{zt}{2\gz c^2(t)}\,dt;$$
and for $\cP$ as in (\ref{cP}) .
If we instead replace parameters $\gL_{\ell},\gs^2$ by $ z\gz,\gz^{3/2},$ respectively, we have such estimates for the matrices $\fN$, but with
$\Theta_{r_k}$ replaced by
$$\Theta_{r_k}(r,\gs)\df\int_{r_k}^r\frac{1}{c(t)} -\frac{z c(t)}{2\sqrt{\gz} t^2}\,dt$$
and for $\cP$ as in (\ref{nextcP}).  We are ready to state
\begin{thm}
Solutions $\vec{X}$ whose components $\Phi$ satisfy $\fS_{\cL}(\vec{\theta})$ (fixed) under 
the following parametrizations for sufficiently large $\gs$ and/or sufficiently large $z$
are $\cC$-approximable with $\vec{Y},\vec{Y}_0$ having the following type of boundary agreements on $[a,b]:$
\begin{itemize}
\item[i)] Parametrization
$\gL_{\ell}\rightarrow$ $\gz^2,$ $\gs^2\rightarrow$ $z\gz$
boundary agreement type $\Pi_{j,k};$
 or, 
\item[ii)] Parametrization $\gL_{\ell}\rightarrow$ $z\gz^{\ga},$  $\gs^2\rightarrow$ $\gz^2$ for 
some $1<\ga<2$ and boundary agreement of type $\Pi(\gs).$
\end{itemize}
\end{thm}
\begin{proof} We need only to modify the proofs of the corresponding theorems along with the supporting propositions.
For part i) it follows, mutatis-mutandis, that $\gamma$ $=O(\gz^2)$ and that, as in Proposition \ref{prop1}, the associated $F_{\gz}$ satisfies
 $\|F\|_{\infty}$ $=O(\gz^{-1}),$ with supremum norm over the square $[a,b]\times[a,b].$ It follows likewise that, under the given boundary conditions, 
 $\fz(\gz)$ $=o(1)$ and, hence, $\cZ(\gz)=$ $o(\gz^{-1}).$

For part ii) we first show that Proposition \ref{innerprodest} if, via (\ref{propeqn1}) , we replace
$\frac{\gs\pi}{4L}$ of (\ref{propeqn2}) by a multiple  $\epsilon \frac{\gs\pi}{4L}$ for some constant $0<\epsilon<1$ for sufficiently large $\gz.$
From there the rest of the proof also follows by directly applying Theorem \ref{thm1} to modify the proof of Theorem \ref{thm3} 
and thereby show that $\gamma$ $>Cz\gs^{\ga}$:  The desired estimate for $\fz$ follows and,  via (\ref{Ypest}), so does that of $\cZ$.
\end{proof}
In more general cases in which one of $\gL_{\ell}$ or $\gs^2$ dominates upon such reparametrization 
one can resort to WKB approximations (cf \cite{o}) rather than  formal Laurent series expansions whereas only leading-order asymptotic estimates are required.
However, a complete treatment of such cases and related the full-system asymptotics are beyond the intended scope of the present work, whereby we now
depart from the adiabatic-equilibrium case.

\section{Non--Adiabatic Equilibrium}\label{s3}

With $N^2>0$ we now proceed with methods parallel to those of Theorems \ref{thm3} and \ref{limitwca} under the parametrization $\gL_{\ell}$ $\rightarrow$ $\gz^2$ (as $\gz$ $\rightarrow$ $\infty$). In the fixed-frequency case, however,  
we now have to consider two cases, $0< N^2<\gs^2$ and $0<\gs^2< N^2,$ which are distinguished by exponential or oscillatory behavior,  respectively, of the
associated residual solutions $\vec{Y}_0$.

\subsection{High Degree and Fixed Frequency in the Exponential Case} 
We now introduce the following notation:
\begin{equation}\label{newerP}
\cP=\text{diag}(1/\sqrt{\rho \cH},\, \gz^{-1}\sqrt{\cH/\rho })
\end{equation}
\begin{equation}\label{newerTheta}
\Theta_{r_i}(r,\gs)\df\int_{r_i}^r\cH(t,\gs)\,dt\,\text{for}\,\cH(r,\gs)\df r^{-1}\sqrt{1- \gs^{-2}N^2(r)}
\end{equation}
We treat various asymptotic cases where $N^2$ $\neq \gs^2$ on $[a,b]$ whereby our asymptotic estimates hold: Indeed,
such asymptotic properties are already known and classified, albeit, under the presumption that the Cowling approximation is accurate (see \cite{acdk} Section 3.4).

\begin{prop}\label{loexpcase}
Under the aforementioned conditions, the residual equation has fundamental solutions $\fN_{j,k}$ as in (\ref{fN}) but with $\cP$ and $\Theta$ as in (\ref{newerP}) and (\ref{newerTheta}).
\end{prop}
We now proceed as in the corresponding adiabatic equilibrium case, but prepare to apply (\ref{mainineq2}) instead. We show that the integral operator $\cF$ differs from a symmetric operator by
operators of small $L^2$ norm: 
\begin{prop}\label{loexpcaseprop} Corresponding to a basis $\fN$ as in (\ref{fN}) the operator $\cF$ as in (\ref{main})  satisfies 
\begin{equation}\|\cF_{\gz}-\cF_{\gz}^{(s)}\|=O(\gz^{-1})\label{opdiff}\end{equation}
where the operator $\cF_{\gz}^{(s)}$ is symmetric and $L^2-$norm bounded uniformly in $\gz,$ whose kernel $F^{(s)}$ satisfies
\begin{equation}\label{Fs}F_{\gz}^{(s)}(r,t)=O(\gz)\end{equation}
($\text{as}\, \gz\rightarrow\infty$)
on the square $[a,b]\times[a,b].$
\end{prop}
\begin{proof}
We follow the construction (\ref{F})  to find $F_{\gz}=F_{\gz}^{(s)}(1+O(\gz^{-1}))$ for
\begin{equation}\label{Jsbnd}
F_{\gz}^{(s)}(r,t)=\frac{\gk}{\Delta(\gz)}\sum_{j,k=1}^2(\frac{\gz}{\gs^2})^{k-j}A_{j,k}(r,t) F^{(j,k)}_{\gz}(r,t)
\end{equation}
where
\begin{equation}\label{Jjk}
{F}^{(j,k)}_{\gz}(r,t)\df E_j^{\top}\fM(r,\zeta) \fX(r,t)\text{adj}(\fM)(t,\zeta) E_k
\end{equation} and
where $A_{j,k}$ are products of various elements of $(\vec{\cC}^{\top}\cP)(r)$ and  $(\cP^{-1}\vec{\cG})(t),$ 
given by
\begin{align} \label{Ajk} A_{1,2}(r,t)&=-w(r)w(t)\,\text{for} \,w=\frac{N^2r}{g}\sqrt{\frac{\rho}{\cH}}\\
A_{1,1}(r,t)&=u(r)v(t)=-A_{2,2}(t,r)\,\text{for}\, v= h\sqrt{\frac{\cH}{\rho}}\,\text{and}\, u=\frac{N^2}{g}\sqrt{\frac{\rho}{\cH}}\notag\\
A_{2,1}(r,t)&=f(r)f(t)\,\text{for}\,f=h\sqrt{\frac{\cH}{\rho}} \notag
\end{align}
That $F_{\gz}^{(s)}$ is symmetric results from the following: For $j\neq k,$
\begin{align}\notag
{F}^{(j,k)}_{\gz}(r,t)&\df E_j^{\top}\fM(r,\zeta) \fX(r,t)\text{adj}(\fM)(t,\zeta) E_k
\\
 &=E_j^{\top}\fM(r,\zeta) \fX(r,t)\fI\left(\fM(t,\zeta)\right)^{\top}\left(\fI^{\top} E_k\right)\notag\\
&=
\left(E_j^{\top}\fM(r,\zeta)\right) \fI^{\top}\fX(t,r)\left(\fM(t,\zeta)\right)^{\top}\left(\fI^{\top} E_k\right)\notag\\
&=
\left(\fI E_j\right)^{\top}\text{adj}(\fM)(r,\zeta)\fX(t,r)\left(\fM(t,\zeta)\right)^{\top}\left(\fI^{\top} E_k\right)\notag\\
&=
E_k^{\top}\text{adj}(\fM)(r,\zeta)\fX(t,r)\left(\fM(t,\zeta)\right)^{\top} E_j\notag\\
&=\left(F^{(k,j)}_{\gz}(t,r)\right)^{\top}=F^{(k,j)}_{\gz}(t,r);\notag
\end{align}
\begin{align}
{F}^{(j,j)}_{\gz}(r,t)
&=E_j^{\top}\fM(r,\zeta) \fX(r,t)\text{adj}(\fM)(t,\zeta) E_j\notag\\
 &=E_j^{\top}\fM(r,\zeta) \fX(r,t)\fI\left(\fM(t,\zeta)\right)^{\top}\fI^{\top} E_j\notag\\
&=
\left(E_j^{\top}\fM(r,\zeta)\right) \fI^{\top}\fX(t,r)\left(\fM(t,\zeta)\right)^{\top}\fI^{\top} E_j\notag\\
&=
\left(\cI E_j\right)^{\top} \text{adj}(\fM)(r,\zeta)\fX(t,r)\left(\fM(t,\zeta)\right)^{\top}\fI^{\top} E_j\notag\\
&=
-E_k\text{adj}(\fM)(r,\zeta)\fX(t,r)\left(\fM(t,\zeta)\right)^{\top}E_k\notag\\
&=-\left(F^{(k,k)}_{\gz}(t,r)\right)^{\top}=-F^{(k,k)}_{\gz}(t,r)\notag
\end{align}
Finally, (\ref{Fs}) follows as in the proof of Theorem \ref{thm3}; and, (\ref{opdiff}) follows from the
fact that each $F_{\gz}^{(j,k)}(r,t)$ is majorized by the expression in (\ref{kerest}) whose 
$L^2$ norm on the square grows to order $O(\gz^{-1}e^{\gz\Theta_a(b)}).$
\end{proof}

We proceed to show that weak and strong CA hold in
\begin{thm}\label{thm5}
The conclusions of Theorem \ref{thm3} hold for the solution components $\vec{Y},\vec{Y}_0$ and $\Phi$ in the case that $0<N^2<\gs^2$ ($\gs$ fixed)
hold except for (\ref{diffest}) replaced by 
\begin{equation}\notag
\|\eta-\eta_0\|_{\infty}\leq \frac{C_1}{\gz^2}\|\eta_0\|_{\infty},\,\|u-u_0\|_{\infty}\leq \frac{C_2}{\gz}\|\eta_0\|_{\infty}
\end{equation}
\end{thm}
\begin{proof}
Since $\|J^{(s)}\|_{\infty}$ $=O(\gz)$ we conclude that $\gamma$ $=\gz^2(1+O(\gz^{-1}))$ of Theorem \ref{thm1} 
and the results thereby hold to obtain $\fz(\gz)$ $=O(\gz^{-2}).$

The estimates on the $J^{(j,k)}$ in turn result in
$$|E_1^{\top}\cC(r)M_{\gz}(r,t)\vec{\cG}(t)|
=O(\gz);  |E_2^{\top}\cC(r)M_{\gz}(r,t)\vec{\cG}(t)|=O(1)
$$ so that strong CA holds with modulus $\cZ(\gz)$ $=O(\gz^{-1}).$
\end{proof}
\subsection{Fixed  Degree and High Frequency}
We now set $\cP$ and $\Theta$ as in  (\ref{nextcP})
and proceed to construct operators $\cF,\cF^{(s)}$ with the following 
\begin{prop}
There are fundamental solutions $\fN(\cdot,\gs)$ of the residual equation with
those properties of (\ref{os}) and (\ref{nextcP}).
\end{prop}
\begin{proof}
The result follows from Proposition \ref{LWest} (see Appendix A) whereby the leading-order estimates are independent of $N.$
\end{proof}
We are ready to state
\begin{prop}\label{fJests} The operator $\cF_{\gs}$ as in (\ref{calF}) under boundary condtions $\Pi(\gs)$ as in Proposition \ref{prop2} satisfies 
$\|\cF_{\gs}-\cF_{\gs}^{(s)}\|=O(\gs^{-1})$
where $\cF_{\gs}^{(s)}$ is symmetric and where 
$\|\cF_{\gs}^{(s)}\| =O(1)$
as $\gs\rightarrow\infty.$
\end{prop}
\begin{proof}
We compute
$F=F_{\gs}^{(s)}(1+O(\gs^{-1}))$ where
$$F^{(s)}_{\gs}(r,t)=\frac{\gk}{\Delta(\gs)}\sum_{j,k=1}^2\gs^{j-k}A_{j,k}(r,t) F^{(j,k)}(r,t)
$$
for $F^{(j,k)}$ as in (\ref{Jjk}) and for
$A_{j,k}$ as in (\ref{Ajk}) but with $\cH$ replaced by $c/r^2$ whereby the symmetric property of $\cF_{\gs}^{(s)}$ is clear. 
\end{proof}
\begin{thm}
The results of Theorem \ref{limitwca} hold also for the case $N^2>0.$
\end{thm}
\begin{proof}
The proof is the same, mutatis-mutandis, as that of the cited theorem except for the applications of  Proposition \ref{fJests}
and the estimates
$|E_1^{\top}\cC_{\gs}(r)M(r,t)\vec{\cG}(t)|
=O(1);$ $ |E_2^{\top}\cC_{\gs}(r)M(r,t)\vec{\cG}(t)|=O(\gs^{-1}).
$
\end{proof}
\subsection{High Degree and Fixed Frequency in the Oscillatory Case}
We return to the parametrization $\gL_{\ell}\rightarrow \gz^2$ but now 
define $\Theta$ as in (\ref{newerTheta}), $\fN_{j,k}$ as in (\ref{formb1}) and
$\cH$ $\df\frac{1}{r}\sqrt{\frac{N^2}{\gs^2}-1}$
with $\cP$ as in (\ref{newerP}), replacing $\fM_{j,k}$ with 
\begin{equation}\notag\vec{\fM}_{1,k}(r,\gz) \df \begin{bmatrix}\sin(\zeta \Theta_{r_k}(r))\\ \cos(\zeta \Theta_{r_k}(r))\end{bmatrix},\,
 \vec{\fM}_{2,k}(r,\gz) \df \begin{bmatrix}\cos(\zeta \Theta_{r_k}(r))\\- \sin(\zeta \Theta_{r_k}(r))\end{bmatrix}
\end{equation} (associated with internal gravity waves, cf. \cite{svh})
to form bases as in (\ref{formb1}). Here, the corresponding $\gD_{j,k}$ can be bounded away from $0,$ 
as was arranged for those of  (\ref{dets}),  via boundary conditions of type
$\Pi(\gz).$ We can then proceed as in the exponential case to prove
\begin{thm}\label{thm7}  The results of Theorem \ref{thm5} hold for case $0<\gs^2<N^2(r)$ ($\gs$ fixed) on $[a,b]$
except for $\vec{Y},\vec{Y}_0$ having boundary agreement of the form $\Pi(\gz).$
\end{thm}
\begin{proof}
The results of Proposition \ref{loexpcaseprop} hold for $\vec{Y},\vec{Y}_0$ with boundary agreement of the form $\Pi(\gz)$ whereby
construction of the associated $F^{(s)}$ follows likewise with the coefficients $A_{j,k}$ of 
the kernel 
$
F^{(s)}
$
as in (\ref{Jsbnd})
being the same except that each of
$\frac{1}{\Delta(\gz)}$ and the $F^{(j,k)}(r,t)$ are uniformly bounded on the square for sufficiently large $\gz$.
\end{proof}
\section{Some Examples from Asymptotic Estimates: Sharper Weak-CA Estimates }\label{s4}

For a given fundamental solution $X(r,\zeta)$ to the system (\ref{otherfirst})-(\ref{otherthird}), depending on some (large) general parameter $\gz$ we denote by ${\fR}_j$ the corresponding rows of $X,$ with
${\fR}_{3}$ $=(\Phi_1,\hdots,\Phi_4 )$ and ${\fR}_{4 }$  $=(\Phi'_1,\hdots,\Phi_4')$ in particular, and set
\begin{equation}\notag
\cW(\gz)\df\left({\fR}_3^{\top}(a,\gz),{\fR}_3^{\top}(b,\gz),{\fR}_4^{\top}(a,\gz),{\fR}_4^{\top}(b,\gz)\right)^{\top}
\end{equation} 
With $\cW$ being invertible for large $\gz$ we construct solutions via 
$\vec{\cX}=X\cW^{-1}\vec{V}$ with boundary values given by
$$\vec{V}=\left(\Phi(a,\gz),\, \Phi(b,\gz),\,\Phi'(a,\gz),\,\Phi'(b,\gz)\right)^{\top}$$

In the following examples we will simplify much of the calculations by the following methods. We will impose the boundary conditions on $\vec{\Phi}$ in setting
 $\fW\df\det\cW$ and construct a solution
$\vec{X}=X\vec{\fC}$ for 
\begin{equation}\label{simpbc}\vec{V}(\gz)=(\fW(\gz),0,0,0)^{\top}\,{\rm and }\,\vec{\fC}(\gz)\df\cW^{-1}(\gz)\vec{V}(\gz).  \end{equation} 
We will also
resort to the following simple, but useful, estimate arising from a straightforward application of the Cauchy-Schwartz inequality:
\begin{prop}\label{simpprop}
Suppose $f(r,\gz),g(r,\gz)$ satisfy $\frac{\|g(\cdot,\gz)\|}{\|f(\cdot,\gz)\|}$ $\rightarrow 0$ as $\gz\rightarrow \infty.$ Then,
$$\|f+g\|=\|f\|(1+o(1))\, \text{as} \,\gz\rightarrow\infty$$
\end{prop}
\subsection{High-Degree Asymptotics, Non--Adiabatic Equilibrium}
We first consider the exponential case as we apply the estimates Proposition \ref{firstasymp} (see Appendix A) for a fundamental solution $X$ with $\Theta(1)$ $=0$. We compute from (\ref{simpbc})
$$\vec{\fC}
=(I+o(1))\left(A_1\frac{e^{\gz\Theta(R)}R^{\gz}}{\gz^{3}},A_2\frac{1}{\gz^{3}},A_3 \frac{e^{\gz\Theta(R)}R^{\gz}}{\gz^{2}},A_4\frac{e^{-\gz\Theta(R)}R^{\gz}}{\gz^{2}}\right)^{\top}$$
for constants $A_j$ given by the following:
$$A_1=\gs^4\kappa^2\frac{\sqrt{\cH(1)\cH(R)\rho(1)\rho(R)}}{g(1)g(R)\sqrt{R}}\left(1-\frac{1}{R\cH(R)}\right);
$$
$$A_2=\frac{2-\gs^4\kappa^2\rho^{3/2}(R)}{g^2(R)};A_3=\frac{\gs^2\kappa\sqrt{\rho(R)}}{g(R)\sqrt{R\cH(R)}}\left(\cH(R)-\frac{1}{R}\right);
$$
$$A_4=-\gs^2\kappa\frac{\sqrt{\cH(R)\rho(R)}}{g(R)\sqrt{R}}\left(1+\frac{1}{R\cH(R)}\right)
$$
In order to apply Proposition \ref{simpprop} we need to appraise integrals of the type addressed in the following:
\begin{lem}
For smooth functions $f,h$ with $g>0$ on $[1,R]$ the following hold (as $\gz\rightarrow \infty$) for $H(r)\df $ $\int_1^rh(t)\,dt:$ 
$$\int_1^Re^{\gz H(r)}f(r)\,dr=\frac{f(R)}{\gz h(R)}(1+o(1))\,\,
if \,f(R)\neq 0;$$ 
$$\int_1^Re^{-\gz H(r)}f(r)\,dr=\frac{f(1)}{\gz h(1)}(1+o(1))\,\,
\,if f(1)\neq 0.$$
\end{lem}
We thus compare terms of the product ${\fR}_3\vec{\fC}$ in accord with Proposition \ref{simpprop}, to find (modulo $o(1)$)
$$\gz^{8}R^{-2\gz}e^{-2\gz\cH(R)}\|\Phi\|^2\equiv\|-A_1r^{-\gz}+\gs^2\gk A_3e^{-\gz\Theta}\sqrt{\rho}/(g\sqrt{\cH})\|^2
$$
$$\equiv\frac{\kappa^4\gs^8(\cH(1)-1)^2(1+3\cH(1)+\cH^2(1))(1-R\cH(R))^2\rho(1)\rho(R)}{2R^3\gz g^2(1)g^2(R)\cH^2(1)\cH(R)(1+\cH(1))}
$$
We now estimate $\vec{\cC}\cdot\vec{Y}:$ 
We find $\vec{Y} =(I+o(1))\fA\vec{A}$ with 
\begin{equation}\label{Ynaught}
\fA=\begin{bmatrix}u_{1}&u_{2}\\
\eta_{1}&\eta_{2}\end{bmatrix},
\vec{A}\df(A_3 \frac{e^{\gz\Theta(R)}R^{\gz}}{\gz^{2}},A_4\frac{e^{-\gz\Theta(R)}R^{\gz}}{\gz^{2}})^{\top}
\end{equation}
where, for $j=1,2,$
$u_{j}(\cdot,\gz)\df\frac{e^{(-1)^j\gz\Theta}}{\sqrt{\cH\rho}},$ $
\eta_{j}(\cdot,\gz)\df\frac{\sqrt{\cH}e^{(-1)^j\gz\Theta}}{\gz\sqrt{\rho}}
$
so that
$$\|\vec{\cC}\cdot\vec{Y}\|^2=\kappa^2A_3^2 \frac{e^{2\gz\Theta(R)}R^{2\gz}}{\gz^{4}}\|\frac{ N^2\rho}{g}u_1\|^2(1+o(1))
$$
 to find (using $N^2/\gs^2=1-r\cH^2$)
$\|\Phi\|=\fz_0(\gz)\|\vec{\cC}\cdot \vec{Y}\|(1+o(1))$
for $$\fz_0^2(\gz)
=F(\cH(1))/\gz^4\, \text{with}\,
F(x)\df  \frac{(1-x)(1+3x+x^2)}{(1+x)^2}$$
It is easy to show that $F(x)$  is strictly decreasing on $[0,1)$ and thus
$F(x)\leq F(0) =1,$ so that $\fz_0 <\gz^{-2}$ by the condition on $N/\gs$ providing an estimate of $\fz=$ $1/\gamma$ sharper than that following from Theorem \ref{thm5}. 

We address the oscillatory case $N/\gs >1$ and follow the procedure above. From Proposition \ref{loosccase} (Appendix A) we 
have for $X=(\vec{\cX}_1,\vec{\cX}_2,Im\vec{\cX}_3,Re\vec{\cX}_3)$
 $$\vec{\fC}
=(I+o(1))\left(\frac{R^{\gz}B_1(\gz)}{\gz^{3}},\frac{B_2}{\gz^{3}}, \frac{R^{\gz}B_3(\gz)}{\gz^{2}},\frac{R^{\gz}B_4(\gz)}{\gz^{2}}\right)^{\top}$$
for $B_j$ given by $B_2=-A_2$

$$B_1(\gz)=\gs^4\kappa^2\frac{\sqrt{\cH(1)\cH(R)\rho(1)\rho(R)}}{g(1)g(R)\sqrt{R}}\left(\sin(\gz\cH(R))+\frac{\cos(\gz\cH(R))}{R\cH(R)}\right)
$$
$$B_3(\gz)=\gs^2\kappa\frac{\sqrt{\cH(R)\rho(R)}}{g(R)\sqrt{R}}\left(\cos(\gz\cH(R))-\frac{\sin(\gz\cH(R))}{R\cH(R)}\right)
$$
$$B_4(\gz)=\gs^2\kappa\frac{\sqrt{\cH(R)\rho(R)}}{g(R)\sqrt{R}}\left(\sin(\gz\cH(R))+\frac{\cos(\gz\cH(R))}{R\cH(R)}\right)
$$
We now estimate the resulting $\Phi=$ $\fR_3\vec{V}$ and $\vec{\cC}\cdot \vec{Y}$ using Proposition \ref{simpprop}  along with simple trigonometric identities and the 
following for $f,G$ as above:
\begin{equation}\notag\|f\sin^2(\gz G(\cdot))\|^2,\|f\cos^2(\gz G(\cdot))\|^2= \frac{1}{2}\|f\|^2(1+o(1))
\end{equation}
\begin{equation}\notag
\int_1^Rf(r)\sin(\gz G(r)) dr, \int_1^Rf(r)\cos(\gz G(r)) dr,\int_1^Rf(r)r^{\pm\gz}dr=o(1) 
\end{equation}
as $\gz\rightarrow \infty.$ We thus have the norm estimate
$$\|\Phi\|^2=\frac{B_3^2+B_4^2}{2\gz^{8}}\gk^2R^{2\gz}\| \sqrt{\rho}/(g\sqrt{\cH})\|^2(1+o(1))
$$

For comparison, we compute $\vec{Y}$ $=(I+O(\gs^{-1}))\fA\vec{A}$
with $\fA$ as above, but for 
$u_1(\cdot,\gz)=\frac{\cos(\gz\Theta)}{\sqrt{\cH\rho}},$ $u_2(\cdot,\gz)$ $=\frac{\sin(\gz\Theta)}{\sqrt{\cH\rho}}$
$\eta_1(\cdot,\gz)=\frac{\sqrt{\cH}\sin(\gz\Theta)}{\gz\sqrt{\rho}}$ and $\eta_2(\cdot,\gz)=-\frac{\sqrt{\cH}\cos(\gz\Theta)}{\gz\sqrt{\rho}}$
and for $\vec{A}$ $=R^{\gz}\gz^{-2}(B_3,B_4)^{\top}$.
Then 
$$\|\vec{\cC}\cdot\vec{Y}\|^2=\frac{B_3^2+B_4^2}{2R\gz^{8}}\gk^2R^{2\gz}\| N\sqrt{\rho}/(g\sqrt{\cH})\|^2(1+o(1))
$$
whereby the modulus $\fz(\gz)$ satisfies the estimate as above, but with
$\fz_0$ $\leq \max_{[a,b]}\gs^2/N^2$ $<1,$ showing that the inferred estimate of Theorem \ref{thm7} is not sharp.
\subsection{High-Degree Asymptotics, Adiabatic Equilibrium}
We follow (\ref{simpbc})  to compute
$$\vec{\fC}
=\frac{A}{\gz^2}(I+o(1))\left(\frac{-1}{2\gz},\,1,\,\frac{1}{2\gz},\,-1\right)^{\top}$$
where for $A=k^2\gs^2F(R)R^{-2}$ and $F(r)\df$ $\int_1^rt\sqrt{\rho(t)}/c^2(t)dt.$
We apply the following estimates for a smooth function $f,g$ such that $f(1)=f'(1)$ $=g(R)=g'(R)=0$
$$\int_1^Rr^{-2\gz}f(r)dt = \frac{f''(1)}{8\gz^3}(1+o(1))\,
{\rm if} f''(1)\neq 0;$$
$$\int_1^Rr^{2\gz}g(r)dt = \frac{R^{2\gz+3}g''(R)}{8\gz^3}(1+o(1))\,
{\rm if} g''(R)\neq 0$$
as $\gz\rightarrow \infty.$
From these estimates and Proposition \ref{simpprop} we find
\begin{align}\|\Phi\|^2 &=\frac{A^2}{4\gz^8}\|r^{\gz-1/2}(F(\cdot)-F(R))\|^2(1+o(1))\notag\\\notag
&=\frac{A^2R^{2\gz+2}(F'(R))^2}{16\gz^{11}}(1+o(1))
\end{align}

\begin{align}\notag\|\vec{\cC}\cdot\vec{Y}\|^2&=\frac{A^2}{\gz^{6}}\|hr^{\gz-1/2}/\sqrt{\rho}\|^2(1+o(1))\\
&=\frac{A^2R^{2\gz+2}(F'(R))^2}{2\gz^{7}}(1+o(1))
\notag
\end{align}
whereby the modulus $\fz(\gz)$ satisfies $\fz(\gz)=\frac{1}{2\sqrt{2}\gz^2}(1+o(1)).$  Hence
the estimate as in Proposition \ref{frmrthm2} is not sharp in this case.

\subsection{High-Frequency Asymptotics}\label{HFEg}
We note that the estimates of Proposition \ref{lrgsig} (see Appendix B) can be applied to the system over intervals 
$[a,\gb]$ $\subseteq$ $[a,b]$ for $a<\gb \leq b$ to the corresponding determinant 
$\fW_{\gb}(\gs)$  with implied bounds uniform in $\gb.$ We use such estimates
for the case that  $\Theta(a)=0$ and 
take $o(1)$ to hold on an unbounded interval whereby $|\fW_{\gb_{\gs}}(\gs)|$ $>\frac{c}{\gs^2}$ for all sufficiently large $\gs$ $>0.$ 
The estimates that are developed here can be developed in similar manner for alternate choices of bases. We compute norm estimates of $\Phi=\fR_3\vec{\fC}$ and of $\cC_{\gs}\cdot\vec{Y}$ following the procedures above.

For $\Phi_{0,j}$ as in Proposition \ref{LWest}, we set
\begin{align}\notag
\fs(r,\gb)\df&\Phi_{0,1}(\gb)\Phi_{0,2}(r)-\Phi_{0,1}(r)\Phi_{0,2}(\gb);
\\\label{fsfw}
 \fw(r,\gb)\df&\Phi_{0,1}(\gb)\Phi'_{0,2}(r)-\Phi_{0,1}'(r)\Phi_{0,2}(\gb)\end{align}
and for  $S_{\gb,\gs}\df \sin(\gs\Theta(\gb)),$ $C_{\gb,\gs}\df \cos(\gs\Theta(\gb))$,
$f$ $=\frac{\sqrt{\rho}}{r\sqrt{c}},$
we find
$$\vec{\fC}
=(I+o(1))\left(A_1\frac{S_{\gb,\gs}}{\gs^{2}},A_2\frac{S_{\gb,\gs}}{\gs^{2}}
, \frac{A_3C_{\gb,\gs}-A_4}{\gs},A_3\frac{C_{\gb,\gs}}{\gs}\right)^{\top}$$
for
$\notag A_1=$ $\kappa^2f(a)f(\gb)\Phi_2(\gb),$
$A_2=$ $-\kappa^2f(a)f(\gb)\Phi_1(\gb),$
$A_3=$ \newline$\kappa f(\gb)\fw(\gb,\gb)$ and
$A_4=\kappa f(a)\fw(a,\gb)$
 by which we find
$$\Phi=-f(a)f(\gb)\fs(\cdot,\gb)\frac{\kappa^2S_{\gb,\gs}}{\gs^2}(1+o(1)).
 $$
We then estimate $\vec{Y}=(I+O(\gs^{-1}))\fA\vec{A}$ as in (\ref{Ynaught}) but for
for $\vec{A}=$ $\gs^{-1}(A_3C_{\gb,\gs}-A_4,A_3 C_{\gb,\gs})^{\top}$
and for 
$$u_1=\frac{r\sin(\gs\Theta)}{\gs\sqrt{c\rho}},u_2=\frac{r\cos(\gs\Theta)}{\gs\sqrt{c\rho}},\eta_1=\frac{c}{\gs r^2}u_2,\eta_2
=-\frac{c}{\gs r^2}u_1
$$
Setting $B_1(\gb,\gs)=$ $f(\gb)\fw(a,\gb) C_{\gb,\gs}-f(a)\fw(\gb,\gb)$ and
$B_2(\gb,\gs)$ $=$ $-f(\gb)\fw(a,\gb) S_{\gb,\gs},$ we have
$$\vec{\cC}_{\gs}\cdot \vec{Y}=\frac{\kappa^2r^2f}{\gs c}(B_1\cos(\gs\Theta)+B_2\sin(\gs\Theta))(1+o(1))$$
We may then choose $\gb=\gb_{\gs}$ $=a+z\gs^{-1}$
 in which case it is not difficult to show that
$
 B_1(\gb_{\gs},\gs)/S_{\gb_{\gs},\gs} =O(z\gs^{-1})
$
whereby
$\|\vec{\cC}_{\gs}\cdot \vec{Y}_0\|^2=\frac{\kappa^4}{2\gs^4}\|r^2f/c \|^2B_2^2(\gb_{\gs},\gs)(1+o(1)).$
The above results combine to give $\|\Phi(\cdot,\gs)\|_{\gb_{\gs}}=\fz_0(\gs)\|\vec{\cC}_{\gs}\cdot\vec{Y}(\cdot,\gs)\|_{\gb_{\gs}}$ 
for
$$\fz(\gs) =\frac{\sqrt{2}f(a)\|\fs(\cdot,a)\|_{\gb_{\gs}}}{\gs|\fw(a,a)|\|r^2f/c\|_{\gb_{\gs}}} (1+o(1))=O(\gs^{-2})
$$

\section{Appendix A}\label{s5}
\subsection{Full-System Asymptotics}
We provide asymptotic estimates for fundamental solutions for various systems involved in the article using standard methods
such as those of \cite{cl,t,w}.
In order to accommodate special cases, we introduce parameters $\gl,\gm$ to the system (\ref{full}), whereby we replace
$\cG, \cC$ by $\gl \cG, \gm\cC,$ respectively.
In each case we begin the process by applying shear transformations $\vec{W}= \cD \vec{X}$ which yield
$$\cA(r,\gz)\df\cD^{-1}(\gz) A(r,\gz) \cD(\gz) =\gz \cA_0(r) +\cA_1(r)+ \gz^{-1}\cA_2(r),$$
whereby $\cA_0$ is diagonalizable.
From the diagonalization process a transformation $\vec{Y}= P\vec{W}$ further yields an equation of the form
\begin{equation}
\vec{Y}' = (\gz \cB_0(r) +\cB_1(r)+ \gz^{-1}\cB_{2}(r) +O(\gz^{-2}))\vec{Y}
\label{Beqn}
\end{equation}
where $\cB_0=P^{-1}\cA_0P$ is diagonal.
For our purposes we will need only to retain leading-order estimates, with methods and theorems of \cite{cl, t, w} in hand, as the methods applied will depend initially on the form of $\cB_0.$ 
\subsection{High-Degree Asymptotics, Non--Adiabatic Equilibrium}
We may apply Theorem 2.1, Chapter 6 \cite{cl}
since the diagonal elements of $Q'_0(r)=\cA_0(r)$ are always distinct in value on $[a,b].$  Set  $\cB_1=\cA_1-P^{-1}P'$ and
\begin{equation}\label{formalsol}
P_1\cA_0=\cA_0P_1+\cB_1 \,(\text{off the diagonal});
 Q_1'= \cB_1\,(\text{on the diagonal})
\end{equation}
\begin{equation}\label{P}\cD=\begin{bmatrix}
1&0&0&0\\
0&\frac{1}{\gz}&0&0\\
0&0&\frac{1}{\gz}&0\\
0&0&0&1\end{bmatrix};
P=\begin{bmatrix}
0&0&\frac{-1}{\cH}&\frac{1}{\cH}\\0&0&1&1\\-r&r&0&0\\1&1&0&0
\end{bmatrix};
\end{equation}
$$
Q_0=\begin{bmatrix}-\ln r&0&0&0\\
0&\ln r&0&0\\
0&0&-\Theta&0\\
0&0&0&\Theta\end{bmatrix};
Q_1=\begin{bmatrix}\ln r^{-\frac{3}{2}}&0&0&0\\
0&\ln r^{-\frac{3}{2}}&0&0\\
0&0&\frac{-\ln F }{2}&0\\
0&0&0&\frac{-\ln F }{2}\end{bmatrix}
$$
with $F(r)\df r\rho(r)/\cH(r)$ and $\Theta$ as in (\ref{cP}). 
 
\begin{prop}\label{firstasymp}The ES system (\ref{otherfirst})-(\ref{otherthird}) with $\gL_{\ell}$ replaced by the continuous parameter $\gz$ has a fundamental solution satisfying $X(r,\gz)=(I +O(\gz^{-1}))\cX(r,\gz)$ as $\gz\rightarrow \infty$ for $\cX=$
\begin{equation}\label{firstmtx}
\left(
\begin{array}{cccc}
 \frac{\gl r^{\frac{3}{2}-\zeta }}{\zeta  g} & -\frac{\gl r^{\zeta +\frac{3}{2}}}{\zeta  g} & -\frac{e^{-\zeta \Theta}}{\sqrt{\rho\cH}} & \frac{ e^{\zeta  \Theta}}{\sqrt{\rho\cH}} \\
 -\frac{\gl r^{\frac{1}{2}-\zeta }}{\zeta ^2 g} & -\frac{\gl r^{\zeta +\frac{1}{2}}}{\zeta ^2 g} & \frac{ e^{-\zeta  \Theta}\sqrt{\cH}}{\zeta \sqrt{\rho}} & \frac{e^{\zeta  \Theta}\sqrt{\cH}}{\zeta\sqrt{\rho} }
\\
 -\frac{r^{-\zeta -\frac{1}{2}}}{\zeta } & \frac{r^{\zeta -\frac{1}{2}}}{\zeta } & \frac{\mu  \gk \gs^2 \sqrt{ \rho}  e^{-\zeta  \Theta}}{\zeta ^2 g \sqrt{\cH}} & \frac{-\mu  \gk \gs^2 \sqrt{\rho}  e^{\zeta  \Theta}}{\zeta ^2 g \sqrt{\cH}} \\
 r^{-\zeta -\frac{3}{2}} & r^{\zeta -\frac{3}{2}} & -\frac{\gk \mu  \gs^2 \sqrt{\rho\cH}  e^{-\zeta \Theta}}{\zeta  g} & -\frac{\gk \mu  \gs^2 \sqrt{\rho\cH}  e^{\zeta \Theta}}{\zeta  g} \\
\end{array}
\right)
\end{equation}
provided that $0<N^2(r)/\gs^2<1$ on $[a,b].$
\end{prop}
\begin{proof}
The result will follow readily from Theorem 3.1, Chapter 6 of \cite{cl} once we identify leading-order terms
of $\cD PY$ where $Y$ is a formal solution of (\ref{Beqn}).
Write $\fM\df P\tilde{P}_1$ in the form  $\fM=(\fm_{j,k})_{_{1\leq j,k \leq 2}}$ for $2\times 2$ blocks $\fm_{j,k}.$ Then, by inspecting the corresponding block form for $P,$ we find that the leading-order estimates of
 $P(I+\gz^{-1}\tilde{P}_1)$ will be obvious once those of the on-diagonal blocks $\fm_{j,j}$ are ascertained.
To do so, we use (\ref{formalsol}) and note that
$P B_0 P_1= A_0 \fM,$ that
$P\cB_1=A_1P-\derr{}P,$ and that 
the associated on-diagonal blocks of $P(\cB_1-\tilde{\cB}_1)$ and of $\derr{}P$ are all zero. We solve
\begin{equation}\notag
\fm_{1,1} \begin{bmatrix}\frac{-1}{r}&0\\0&\frac{1}{r}\end{bmatrix}-\begin{bmatrix}0&1\\\cH^2&0\end{bmatrix}\fm_{1,1}=\gl
\begin{bmatrix}0&0\\\frac{rN^2}{\gs^2g}&\frac{-rN^2}{\gs^2g}\end{bmatrix}
\end{equation}
\begin{equation}\notag\fm_{2,2} \begin{bmatrix}-\cH&0\\0&\cH\end{bmatrix}-\begin{bmatrix}0&1\\\frac{1}{r^2}&0\end{bmatrix}\fm_{2,2}=\gm
\begin{bmatrix}0&0\\\frac{-\gk N^2\rho}{gr^2\cH}&\frac{\gk N^2\rho}{gr^2\cH}\end{bmatrix}
\end{equation}
to find, using $r^2\cH^2= 1 - N^2/\gs^2,$ 
\begin{equation}\label{m22} \fm_{1,1}=\gl\begin{bmatrix}\frac{r^3}{g}&\frac{-r^3}{g}\\
\frac{-r^2}{g}&\frac{-r^2}{g}
\end{bmatrix};
\fm_{2,2}=\gm\begin{bmatrix}\frac{\gk\gs^2\rho}{g\cH}&\frac{-\gk\gs^2\rho}{g\cH}\\
\frac{-\gk \gs^2\rho}{g}&\frac{-\gk \gs^2\rho}{g}
\end{bmatrix}
\end{equation}
The desired result thus obtains from those leading-order terms of 
$\cX= \cD(P +\gz^{-1}\fM) \exp(\gz Q_0 +Q_1).$
\end{proof}

The procedure for the oscillatory (g-mode) case is almost identical in the following, exept that we redefine $\cH$ by the real-valued
$$\cH(r)\df\frac{1}{r}\sqrt{N^2(r)/\gs^2-1}.$$
\begin{prop}\label{loosccase}
The system as in Proposition \ref{firstasymp} with $N^2(r)/\gs^2 >1$ on $[a,b]$ has a fundamental solution as in (\ref{firstmtx}) except for columns $\vec{\cX}_3,\vec{\cX}_4$ of $\cX$ replaced by
\begin{align}\notag
\vec{\cX}^{\,\top}_3&=\left(\frac{ie^{-i\zeta \Theta}}{\sqrt{\rho\cH}} , \frac{ e^{-i\zeta  \Theta}\sqrt{\cH}}{\zeta \sqrt{\rho}} , \frac{-i\gk \mu  \gs^2 \sqrt{\rho}  e^{-i\zeta \Theta}}{\zeta^2  g\sqrt{\cH}},\frac{-\gk \mu  \gs^2 \sqrt{\rho\cH}  e^{-i\zeta \Theta}}{\zeta  g} \right)
\\\notag
\vec{\cX}^{\,\top}_4&=\left( \frac{ -ie^{i\zeta  \Theta}}{\sqrt{\rho\cH}} , \frac{e^{i\zeta  \Theta}\sqrt{\cH}}{\zeta\sqrt{\rho} }, \frac{i\mu  \gk \gs^2 \sqrt{\rho}  e^{i\zeta  \Theta}}{\zeta ^2 g \sqrt{\cH}} , -\frac{\gk \mu  \gs^2 \sqrt{\rho\cH}  e^{i\zeta \Theta}}{\zeta  g}\right )
\end{align}
\end{prop}
\begin{proof}
Upon changing the definition of $\cH,$ the procedures as those of (\ref{P}) and (\ref{m22}) give
$$ 
P=\begin{bmatrix}
0&0&\frac{i}{\cH}&\frac{-i}{\cH}\\0&0&1&1\\-r&r&0&0\\1&1&0&0
\end{bmatrix},\, 
\fm_{2,2}=\gm\begin{bmatrix}\frac{-i\gk\gs^2\rho}{g\cH}&\frac{i\gk\gs^2\rho}{g\cH}\\
\frac{-\gk \gs^2\rho}{g}&\frac{-\gk \gs^2\rho}{g}
\end{bmatrix}
$$
whereas $P^{-1}P'$ and, hence, $Q_1$ remain the same under the new definition of $\cH$ . The result follows from the given substitutions.
\end{proof}
\subsection{High-Degree Asymptotics, Adiabatic Equilibrium}

For $\Theta$ as in (\ref{cP}) and for $\cD$ as above we arrive at an equation of the form (\ref{Beqn}) with $\cB_0 =Q'_0$ for
$$
P=\begin{bmatrix}
r&0&-r&0\\-1&0&-1&0\\0&r&0&-r\\0&-1&0&-1
\end{bmatrix};
Q_0=\begin{bmatrix}-\ln r&0&0&0\\
0&-\ln r&0&0\\
0&0&\ln r&0\\
0&0&0&\ln r\end{bmatrix},
$$ noting that $Q_0'$ does not have distinct diagonal elements but that it lends itself to an alternate method:
From \cite{t} (see Section 5), we may find matrices $\cQ_1(r)$ and $\cQ_2(r)$ so that the transformation 
\begin{equation}\label{sep}
\vec{X}=\cD P(I+\gz^{-1}\cQ_1)(I+\gz^{-2}\cQ_2)\vec{Y}
\end{equation}
yields an equation of the form (\ref{Beqn}) but
with $\cB_1$ and $\cB_2$  which in $2\times 2$ block form has $0$ blocks off the diagonal.
In this process we compute a matrix $\cQ_1$ whereby $\cB_1$ happens to be diagonal: These are given by
$$
\cQ_1=\frac{1}{4}\begin{bmatrix}0&0&rf&0\\
0&0&0&-1\\
-rf&0&0&0\\
0&1&0&0\end{bmatrix};
\cB_1=-\frac{1}{2}
\begin{bmatrix}
f&0&0&0\\
0&3/r&0&0\\
0&0&f&0\\
0&0&0&3/r\end{bmatrix}$$
for $f(r)$ $\df\derr{}\ln(r\rho(r)).$ Thereafter, with a combination of singular and non-singular perturbation methods (involving non-positive powers of $\gz$), sufficient numbers of terms of a formal solution can be computed directly - at, however, the cost of digressing into procedures beyond the scope of this article. In the present case, we take advantage of the separated form in the following 
\begin{prop}\label{secondasymp}
The system as in Proposition \ref{firstasymp} with $N(r)= 0$ on $[a,b]$ has a fundamental solution $X$ satisfying $X=(I+O(\gz^{-1}))\cX$ as $\gz\rightarrow\infty$ for $\cX=$  
\begin{equation}\label{AdEqSol}
\left(
\begin{array}{cccc}
 -\frac{\lambda  f r^{\frac{1}{2}-\zeta }}{\zeta  \sqrt{\rho }} & -\frac{r^{\frac{1}{2}-\zeta }}{\sqrt{\rho}} & \frac{\lambda  f r^{\zeta +\frac{1}{2}}}{ \zeta  \sqrt{\rho }} & \frac{r^{\zeta +\frac{1}{2}}}{\sqrt{\rho }} \\
 \frac{\lambda  f r^{-\zeta -\frac{1}{2}}}{\zeta ^2 \sqrt{\rho }} & \frac{r^{-\zeta -\frac{1}{2}}}{\zeta  \sqrt{\rho }} & \frac{\lambda  f r^{\zeta -\frac{1}{2}}}{\zeta ^2 \sqrt{\rho }} & \frac{r^{\zeta -\frac{1}{2}}}{\zeta  \sqrt{\rho }} \\
 -\frac{r^{-\zeta -\frac{1}{2}}}{\zeta } & -\frac{\gk \mu  \gs^2 f r^{-\zeta -\frac{1}{2}}}{\zeta ^2} & \frac{r^{\zeta -\frac{1}{2}}}{\zeta } & \frac{\gk \mu  \gs^2 f r^{\zeta -\frac{1}{2}}}{\zeta ^2} \\
 r^{-\zeta -\frac{3}{2}} & \frac{\gk \mu  \gs^2 f r^{-\zeta -\frac{3}{2}}}{ \zeta } & r^{\zeta -\frac{3}{2}} & \frac{\gk \mu  \gs^2 f r^{\zeta -\frac{3}{2}}}{ \zeta } \\
\end{array}
\right)
\end{equation}
where $f(r)$ satisfies $f'$ $=\frac{1}{2}r\sqrt{\rho}/c^2.$ 
\end{prop}
This result is also the subject of another article by the author which is obtained by direct calculations - singular and non-singular perturbation - via symbolic computation \cite{wi}.
\begin{proof}
From the discussion above it is clear that  (\ref{Beqn}) has a formal solution $\cY$  satisfying
$$\cY\exp(-\gz Q_0 -Q_1)= I+ O(\gz^{-1}),$$ 
where $Q_1'=\cB_1:$ By Theorem 6 of \cite{t}, moreover, there is a fundamental solution $Y=$
$(\vec{Y}_1,\vec{Y}_2,\vec{Y}_3,\vec{Y}_4)$ 
satisfying this estimate. 
Since the estimate is independent of both $\gl$ and $\gm,$ we may apply the estimate $\cD\cY$ to  (\ref{full}) in both the separated case $\gm=\gl=0,$ whereby (\ref{AdEqSol}) now holds, and the case $\gm=\gl=1:$ We are thus able to complete the later case via successive approximations to compute

\begin{equation}\notag
\fM_0\df\begin{bmatrix}u_{0,1}&u_{0,2}\\\eta_{0,1}&\eta_{0,2}
\end{bmatrix}=(I+O(\gz^{-1}))
\begin{bmatrix}
-\frac{r^{\frac{1}{2}-\zeta }}{\sqrt{\rho}} &  \frac{r^{\zeta +\frac{1}{2}}}{\sqrt{\rho }} \\
 \frac{r^{-\zeta -\frac{1}{2}}}{\zeta  \sqrt{\rho }} &   \frac{r^{\zeta -\frac{1}{2}}}{\zeta  \sqrt{\rho }} 
\end{bmatrix}
\end{equation}
\begin{equation}\label{fO0}
\fO_0\df\begin{bmatrix} \Phi_{0,1}&\Phi_{0,2}\\\Phi_{0,1}'&\Phi_{0,2}'
\end{bmatrix}=(I+O(\gz^{-1}))
\begin{bmatrix}
 -\frac{r^{-\zeta -\frac{1}{2}}}{\zeta }&\frac{r^{\zeta -\frac{1}{2}}}{\zeta }\\ r^{-\zeta -\frac{3}{2}} &  r^{\zeta -\frac{3}{2}}
\end{bmatrix}
\end{equation}
$$
\fM_0^{-1}\cG_{\gs}\Phi_{0,1}\, , \, -r^{-2\gz}\fM_0^{-1}\cG_{\gs}\Phi_{0,2}=\frac{1}{\gz}(I+O(\gz^{-1}))
\begin{bmatrix}f'\\-r^{-2\gz}f'
\end{bmatrix}
$$
We may use the following estimates
$$\int_r^bt^{-2\gz}f'(t)\,dt =r^{-2\gz}O(\gz^{-1});\,\int_a^rt^{2\gz}f'(t)\,dt =r^{2\gz}O(\gz^{-1})
$$
to find
\begin{align}\notag
u_1(r,\gz)&=\gl u_{0,1}(r,\gz)\frac{1}{\gz}\int_b^rf'(t)\,dt(1+O(\gz^{-1}))\\\notag&+\gl u_{0,2}(r,\gz)\frac{1}{\gz}\int_r^bt^{-2\gz}f'(t)\,dt(1+O(\gz^{-1}))\\\notag
&=\gl u_{0,1}(r,\gz)\frac{1}{\gz}\int_b^rf'(t)\,dt(1+O(\gz^{-1}))
\end{align}
Likewise, we have
$$\eta_1(r,\gz)=\gl \eta_{0,1}(r,\gz)\frac{1}{\gz}\int_b^rf'(t,\gz)\,dt(1+O(\gz^{-1}))
$$ which completes the asymptotic estimates of $\vec{Y}_1.$

We further appraise a solution $\vec{Y}_3$ with
\begin{align}\notag
u_3(r,\gz)=&-\gl u_{0,1}(r,\gz)\frac{1}{\gz}\int_a^rt^{2\gz}f'(t)\,dt(1+O(\gz^{-1}))\notag\\&+\gl u_{0,2}(r,\gz)\frac{1}{\gz}\int_a^rf'(t)\,dt(1+O(\gz^{-1}))\notag\\\notag
=&\gl u_{0,2}(r,\gz)\frac{1}{\gz}\int_a^rf'(t)\,dt(1+O(\gz^{-1})\notag
\end{align}
and, in turn,
\begin{equation}\eta_3(r,\gz)=\gl \eta_{0,2}(r,\gz)\frac{1}{\gz}\int_a^rf'(t)\,dt(1+O(\gz^{-1}))\notag
\end{equation}

We appraise remaining components of solutions $\vec{Y}_2$ and $\vec{Y}_4:$ Using
$$
\fO_0^{-1}\begin{bmatrix}0\\ \gs^2\gk\gm h r^{-2}\eta_{0,1}\end{bmatrix}=\frac{1}{\gz}(I+O(\gz^{-1}))
\begin{bmatrix}f'\\f'r^{-2\gz}
\end{bmatrix}
$$
$$
\fO_0^{-1}\begin{bmatrix}0\\ \gs^2\gk\gm h r^{-2}\eta_{0,2}\end{bmatrix}=\frac{1}{\gz}(I+O(\gz^{-1}))
\begin{bmatrix}f'r^{2\gz}\\f'
\end{bmatrix},
$$
we find for derivative orders $k=0,1$
\begin{align}\notag
\Phi_2^{(k)}(r,\gz)&= \gs^2\gk\gm\Phi^{(k)}_{0,1}(r,\gz)\frac{1}{\gz}\int_b^rf'(t)\,dt(1+O(\gz^{-1}))\\\notag
\Phi_4^{(k)}(r,\gz)&= \gs^2\gk\gm\Phi^{(k)}_{0,2}(r,\gz)\frac{1}{\gz}\int_a^rf'(t)\,dt(1+O(\gz^{-1})).
\end{align}
Our solutions $\vec{Y}_j$ may be modified to involve various antiderivatives of $f';$ indeed, we may suppose, for instance, $f(r)=$ $ \int_a^rf'(t)dt,$ giving (\ref{AdEqSol}).
\end{proof}

\subsection{High-Frequency Asymptotics, Fixed Degree} 

We find in the LW formulation of the system, from the
variable change $\eta=$ $\gs^{-2}(y_{_{LW}}-\Phi),$ that the leading-order matrices take interesting forms which can be exploited. The change of variables results in 
$\vec{X}'=A_{LW}\vec{X}$ for 
\begin{equation}\notag A_{LW}\df
\begin{bmatrix}
\frac{g}{c^2}&\frac{\gL}{\gs^2}-\frac{r^2}{c^2}&\gl\frac{\gL^2}{\gs^2}&0\\
\frac{\gs^2-N^2}{r^2}&\frac{N^2}{g}&0&\gl\\
0&0&0&1\\
\gm\frac{\kappa N^2\rho}{r^2g}&\gm\frac{\kappa \rho}{c^2}&\frac{\gL}{r^2}&\frac{-2}{r}
\end{bmatrix}
\end{equation}
We then  transform the resulting system as in (\ref{sep}) with
$$\cD=\begin{bmatrix}
\frac{1}{\gs}&0&0&0\\0&1&0&0\\0&0&1&0\\0&0&0&1
\end{bmatrix};
P=\begin{bmatrix}
0&0&\frac{-ir^2}{c}&\frac{ir^2}{c}\\
0&0&1&1\\
0&1&0&0\\
1&0&0&0
\end{bmatrix};
$$
$$
\cQ_1=\begin{bmatrix}
0&0&\frac{i\gm\gk \rho}{c}&\frac{i\gm\gk \rho}{c}\\
0&0&0&0\\
\frac{i\gl c}{2}&0&0&0\\
\frac{-i\gl c}{2}&0&0&0
\end{bmatrix};
\cQ_2=\begin{bmatrix}
0&0&0&0\\
0&0&-\gm \gk\rho&-\gm \gk\rho\\
\frac{\gl g}{2}&0&0&0\\
\frac{\gl g}{2}&0&0&0
\end{bmatrix}
$$
to obtain an equation of the form (\ref{Beqn}) for $\gz=\gs:$ We retain low-order terms (in $\frac{1}{\gs}$) on the block diagonal in setting
$$\gs\cB_0(r)+\cB_1(r)=\tilde{\cB}(r,\gs) + \, higher\, order$$ 
to find that the equation $\vec{Y}'=\tilde{\cB}\vec{Y}$ has an exact fundamental solution $\cY$ where 
$$\tilde{\cB}=\begin{bmatrix}
\frac{-2}{r}&\frac{\gL}{r^2}&0&0\\1&0&0&0\\0&0&\frac{-i\gs}{c}&0\\0&0&0&\frac{i\gs}{c}
\end{bmatrix};
\cY=\begin{bmatrix}
\derr{\Phi_{0,1}}&\derr{\Phi_{0,2}}&0&0\\\Phi_{0,1}&\Phi_{0,2}&0&0\\0&0&e^{-i\Theta}&0\\0&0&0&e^{i\Theta}
\end{bmatrix}
$$
for $\Phi_{0,1}(r)=r^{-\ell-1},$ $\Phi_{0,2}(r)=r^{\ell}.$ 
We are ready to state
\begin{prop}\label{LWest} The system 
$\vec{X}'=$ $A_{LW}\vec{X}$ 
has a fundamental solution $X_{LW}$ satisfying $X_{LW}=(I+O(\gs^{-1}))\cX$ as $\gs\rightarrow\infty$ for $\cX=$  
 $$
\left(  
\begin{array}{cccc}
 \frac{\gl  r^2\Phi_{0,1}'}{\gs^2} & \frac{\gl  r^2 \Phi_{0,2}'}{\gs ^2} & -\frac{i r e^{-i \gs  \Theta }}{\gs \sqrt{c\rho }} & \frac{i r e^{i \gs  \Theta }}{\gs  \sqrt{c\rho}} \\
 \frac{\gl  g \Phi_{0,1}'}{\gs ^2} & \frac{\gl  g \Phi_{0,2}'}{\gs^2} & \frac{\sqrt{c} e^{-i \gs  \Theta }}{r \sqrt{\rho }} & \frac{\sqrt{c} e^{i \gs  \Theta }}{r \sqrt{\rho }} \\
\Phi_{0,1} &\Phi_{0,2} & -\frac{\gk \gm  \sqrt{c\rho } e^{-i \gs \Theta }}{\gs^2 r} & -\frac{\gk \gm  \sqrt{c\rho } e^{i \gs\Theta }}{\gs^2 r} \\
\Phi_{0,1}' & \Phi_{0,2}' & \frac{i \gk \gm  \sqrt{\rho } e^{-i \gs \Theta }}{\gs r \sqrt{c}} & -\frac{i \gk \gm  \sqrt{\rho} e^{i \gs \Theta }}{\gs  r \sqrt{c}} \\
\end{array}
\right),$$
for $\Phi_{0,j}$ as above, and for $\Theta$ as in (\ref{nextcP}). 
\end{prop}
\begin{proof}
Standard arguments \cite{t} show that we only need leading-order terms of the formal solution as discussed above:
These are given by the product $\cD P(I+\gs^{-1}\cQ_1+\gs^{-2}\cQ_2)\cY.$
\end{proof}
We note that the leading-order estimates are independent of $N$ and hence apply for both the adiabatic and non-adiabatic cases.
Then, from the associated variable change in the case $\gm=\gl=0,$ we have 
\begin{cor}\label{LWccor}
The residual equation $\cA_{\gs}\vec{Y}=$ $\vec{Y}'$ has a fundamental solution satisfying $Y=$ $(I+O(\gs^{-1}))\cY_0$ for
constants $\vartheta_j$ where
\begin{equation}\label{lrgsigmacor}
\cY_0=
\begin{bmatrix}\frac{r\sin(\gs \Theta+\vartheta_1)}{\gs\sqrt{c\rho}}&\frac{-r\cos(\gs\Theta+\vartheta_2)}{\gs\sqrt{c\rho}}
\\
\frac{\sqrt{c}\cos(\gs\Theta+\vartheta_1)}{ r\sqrt{\rho}}&\frac{\sqrt{c} \sin(\gs\Theta+\vartheta_2)}{r\sqrt{\rho} }
\end{bmatrix} \,\, \text {and}\,\, \cos(\vartheta_1-\vartheta_2)\neq 0
\end{equation}
\end{cor}
\begin{cor}\label{altfundcor}
The system (\ref{otherfirst})-(\ref{otherthird}) has 
fundamental solutions satisfying (the block form)
$$X=\begin{bmatrix}O(\gs^{-2})&(I+O(\gs^{-1}))\cY_0\\(I+O(\gs^{-1}))\fO_0&O(\gs^{-2})\end{bmatrix}$$
for $\fO_0$ as in (\ref{fO0}) and $\cY_0$ as in (\ref{lrgsigmacor}).
\end{cor}
Finally, we end the section with another immediate result in
\begin{cor}\label{CAcor}
No non-zero solution is $\cC$-approximable over a fixed interval $\cI$ as $\gs\rightarrow \infty.$
\end{cor} 
\begin{proof}
After a change of variables, it is clear that, modulo factors of $o(1)$ (as $\gs\rightarrow\infty$), that components of solutions 
$\vec{X}$ satisfy
\begin{align}\notag
\gs^2\eta
&\equiv A_1\Phi_{0,1}+A_2\Phi_{0,2}+A_3\frac{\sqrt{c}}{r\sqrt{\rho}}\cos(\gs\Theta)-A_4\frac{\sqrt{c}}{r\sqrt{\rho}}\sin(\gs\Theta)\\\notag
\Phi
&\equiv A_1\Phi_{0,1}+A_2\Phi_{0,2}-A_3\frac{\kappa\sqrt{c\rho}}{\gs^2r}\cos(\gs\Theta)+A_4\frac{\kappa\sqrt{c\rho}}{\gs^2r}\sin(\gs\Theta)
\end{align}
with coefficients $A_i$ depending only on $\gs.$ Comparing to the basis arising from the $\gl=0$ case, it is not difficult to show that
if $\|\eta-\eta_0\|_{\infty}$ $=o(1)\|\eta\|_{\infty}$ on $[a,b];$ and,
it follows that
$\|\gs^2\eta_0-\Phi\|=\gs^2\|\eta_0\|o(1),$
whereby $\|\Phi\|/\|\eta_0\|$ is unbounded as $\gs\rightarrow\infty.$
\end{proof}
\section{Appendix B}\label{s6}
\subsection{Admissible Multi-Point Boundary Conditions}
Given a fundamental solution $X(r,\zeta)$ to the system (\ref{otherfirst})-(\ref{otherthird}), depending on some large (general) parameter $\gz,$ we will denote by 
 $\vec{\fR}_{3}$ $=(\Phi_1,\hdots,\Phi_4 )$ and $\vec{\fR}_{4 }$  $=(\Phi'_1,\hdots,\Phi_4')$ the corresponding row vectors and let
\begin{equation}\notag
\fW_{a,b}(\gz)\df\left(\vec{\fR}_3^{\top}(a,\gz),\vec{\fR}_3^{\top}(b,\gz),\vec{\fR}_4^{\top}(a,\gz),\vec{\fR}_4^{\top}(b,\gz)\right)
\end{equation}
Then, elementary arguments confirm that if  $\det\fW_{a,b}(\gz)\neq 0,$ then given any $A_{j,k},$ a solution $\vec{X}$ exists such that the components $\vec{\Phi}$ satisfy
 \begin{equation}\notag\begin{bmatrix}\Phi(a,\gz)&\Phi(b,\gz)\\\Phi'(a,\gz)&\Phi'(b,\gz)
\end{bmatrix}
=
\begin{bmatrix}A_{1,1}&A_{1,2}\\A_{2,1}&A_{2,2}
\end{bmatrix}
\end{equation}
It then follows that specific boundary conditions on solution components $\vec{\Phi}$ more specific than those of 
$\fS_{\cL}(\vec{\theta})$ are admissible.
We will also consider multi-point boundary conditions (in the high-frequency case) via the following:
Denoting $\vec{r}= (r_1,r_2,r_3,r_4)$ for  $r_i\in [a,b]$ we set
\begin{equation}\label{vs}
\fV_{\vec{r}}(\gs)\df\left(\vec{\fR}_1^{\top}(r_1,\gs),\vec{\fR}_2^{\top}(r_2,\gs),\vec{\fR}_3^{\top}(r_3,\gs),\vec{\fR}_4^{\top}(r_4,\gs)\right).
\end{equation}
We thereby verify admissibility of boundary conditions of the form $\Pi(\gs)$ imposed separately on components $\vec{Y}$ and $\vec{\Phi}.$
Finally, for ease of exposition, we will refer to the following which can be verified by elementary means:
\begin{lem}\label{proddiff} Let $r>0$ and $\ga_1,\ga_2$ be fixed real numbers with variable $\epsilon>0.$ Then,
$$q(\epsilon)\df r^{\ga_1}(r+\epsilon)^{\ga_2}-(r+\epsilon)^{\ga_1}r^{\ga_2} =(\ga_2-\ga_1)r^{\ga_1+\ga_2-1}\epsilon+O(\epsilon^2)
$$
(as $\epsilon\rightarrow 0^+$). In particular, if $\ga_1\neq\ga_2,$ $|q(\epsilon)| >0$ for all $\epsilon$ sufficiently small. 
\end{lem}
\subsection{Non--Adiabatic Equilibrium, High Degree}
We begin with the exponential case and proceed to examine fundamental solutions from Appendix A on intervals $[1,R]$ for arbitrary (finite) $R>1.$ 
Such cases do not result in loss of generality, considering scaling arguments (or certain homologous transforms. cf. \cite{acdk}): 
\begin{prop}\label{cWprop}  Fundamental solutions as in Proposition \ref{firstasymp}
satisfy\begin{equation}\label{detfW}
 \det{\fW_{1,R}(\gz)}\neq 0 \,\,\text{\rm for all sufficiently large}\, \gz.
\end{equation}  
\end{prop}
\begin{proof}
After row interchange, we find that $|\fW(\gz)|=$ $\det(\cM)(1+o(1))$ where  in block form $\cM\df$ $(M_{jk})$
with invertible blocks $M_{jk}$ given by
$$
M_{11}=\begin{bmatrix}\frac{-1}{\gz}&\frac{1}{\gz}\\1&1\end{bmatrix}, 
M_{22}=\gs^2\gk\begin{bmatrix}\frac{e^{-\gz\Theta(R)}\sqrt{\rho(R)}}{\gz^2g(R)\sqrt{\cH(R)}}&\frac{-e^{\gz\Theta[R]}\sqrt{\rho(R)}}{\gz^2g(R)\sqrt{\cH(R)}}\\
\frac{-e^{-\gz\Theta(R)}\sqrt{\rho(R)\cH(R)}}{\gz g(R)}&\frac{-e^{\gz\Theta(R)}\sqrt{\rho(R)\cH(R)}}{\gz g(R)}
\end{bmatrix},
$$
$$M_{12}=\gs^2\gk\begin{bmatrix}\frac{\sqrt{\rho(1)}}{\gz^2g(1)\sqrt{\cH(1)}}&\frac{-\sqrt{\rho(1)}}{\gz^2g(1)\sqrt{\cH(1)}}\\\frac{-\sqrt{\rho(1)\cH(1)}}{\gz g(1)}&\frac{-\sqrt{\rho(1)\cH(1)}}{\gz g(1)}\end{bmatrix},
M_{21}=\begin{bmatrix}\frac{-R^{-\gz-1/2}}{\gz}&\frac{R^{\gz-1/2}}{\gz}\\R^{-\gz-3/2}&R^{\gz-3/2}
\end{bmatrix}
$$
Here
$\det(\cM)=\det( M_{11} M_{21})\det\left(M_{21}^{-1}M_{22}-M_{11}^{-1}M_{12}\right)
$
where we find that $M_{21}^{-1}M_{22}-M_{11}^{-1}M_{12}=$ $M(I+o(1))$ with
$$\cM=\frac{-k\gs^2}{2\gz}\begin{bmatrix}R^{\gz}e^{-\gz\Theta(R)}a_{11}&R^{\gz}e^{\gz\Theta(R)}a_{12}\\a_{21}&a_{22}
\end{bmatrix}
\text{for}
$$
\begin{align}
a_{11}=\fa(R)(\cH(R)+1),\,& a_{1,2}=\fa(R)(\cH(R)-1),\notag\\a_{21}=\fa(1)(\cH(1)-1),\,&a_{22}=\fa(1)(\cH(1)+1)\notag
\end{align}
and $\fa\df\frac{\sqrt{\rho}}{ g\sqrt{\cH}}.$
It is then clear that $\det \cM$ $\neq 0$ for all sufficiently large $\gz,$ and the conclusion follows.
\end{proof}

We now address the oscillatory case in
\begin{prop} 
Fundamental solutions as in Proposition \ref{loosccase} satisfy (\ref{detfW}).
\end{prop}
\begin{proof}
For $M_{11}$ and $M_{21}$ as in Proposition \ref{cWprop} 
and for $r_1=1,r_2=R$
$$
M_{j2}=\gs^2\gk\begin{bmatrix}\frac{A(r_j)\cos(\gz\Theta(r_j))}{\gz^2}&\frac{-A(r_j)\sin(\gz\Theta(r_j))}{\gz^2}\\
\frac{-B(r_j)\sin(\gz\Theta(r_j))}{\gz}&\frac{-B(r_j)\cos(\gz\Theta(r_j))}{\gz}\end{bmatrix}
$$
for $A=\frac{\sqrt{\rho}}{rg}\frac{1}{\sqrt{\cH}}$ and  $B=\frac{\sqrt{\rho}}{g}\sqrt{\cH},$ leaving  $\Theta$ indefinite. 
We likewise estimate $\det(M_{11}^{-1})\det((M_{jk}))=$
 $$ \frac{\gs^4\gk^2R^{\gz-1/2}}{2\gz^3}\big[\left(A(R)B(1)+A(1)B(R)\right)\cos(\gz(\Theta(R)-\Theta(1)))$$
$$+\left(B(1)B(R)-A(R)A(1)\right)\sin(\gz(\Theta(R)-\Theta(1)))\big](1+o(1))
$$ 
By symmetry arguments, we find that if the given leading-order expression vanishes, it does not  vanish upon interchange $r_1\leftrightarrow r_2.$ 
Therefore, the conclusion follows by judicious choice of $\Theta$ (via arbitrary constant)
depending on $\gz,$ with the corresponding column operations being clear.  
\end{proof}
\subsection{Adiabatic Equilibrium, High Degree}
We return to the $N= 0$ case with 
\begin{prop}\label{lastprop1}
The fundamental solutions as in Proposition \ref{secondasymp} can be chosen to satisfy (\ref{detfW}). 
\end{prop}
\begin{proof}
After rescaling to the domain $r\in [1,R],$ we may set $F(1)=0$ so that after elementary row or column operations $\fW_{1,R}(\gz)$ is a constant multiple of 
$$\det(\cM(1))\det(\kappa\gm\gs^2\gz^{-1}\cM(R))(1+O(\gz^{-1}))$$ for 
$$\cM(r)\df\begin{bmatrix}r^{-\gz-1/2}\gz^{-1}&r^{\gz-1/2}\gz^{-1}\\
-r^{-\gz-3/2}&r^{\gz-3/2}
\end{bmatrix}
$$
whose determinant does not vanish on $[1,R].$
\end{proof}
\subsection{High-frequency, Fixed  Degree} We now consider the case of $\gs\rightarrow\infty$ with $\ell$ fixed, independent of 
$N\geq 0,$ for $r$ on intervals $\cI_{\gb}\df$ $[a,\gb]$ $\subset$ $[a,b]$ as in Section \ref{HFEg}.
\begin{prop}\label{lrgsig}
Fundamental solutions of Proposition \ref{LWest} 
satisfy $\det\fW_{a,\gb_{\gs}}(\gs)$ $\neq 0$
for appropriate choices of $\gb_{\gs}>a,$ depending on $\gs$ 
sufficiently large: Indeed, it suffices that $\gb_{\gs}=$ $a+z\gs^{-1}$ for sufficiently small $z>0.$
\end{prop}
\begin{proof}
We set $\Theta(a)=0$ and apply clear column operations to
compute 
$\det \fW_{a,\gb} =$
$$\gs^{-2}\det(M_{11})\det(M_{21})
\det(\gs^{-1}M_{21}^{-1}M_{22}-M_{11}^{-1}M_{12})(1+O(\gs^{-1}))$$ 
as $\gs\rightarrow \infty$ where
$$
M_{11}(\gb)=\begin{bmatrix}\Phi_{0,1}(a)&\Phi_{0,2}(a)\\\Phi_{0,1}(\gb)&\Phi_{0,2}(\gb)\end{bmatrix},
M_{21}(\gb)=\begin{bmatrix}\Phi_{0,1}'(a)&\Phi_{0,2}'(a)\\ \Phi_{0,1}'(\gb)&\Phi_{0,2}'(\gb)\end{bmatrix},
$$
$$M_{12}(\gb;\gs)=\gk\begin{bmatrix} -c(a)f(a) &0\\ -c(\gb)f(\gb)\cos(\gs\Theta(\gb))&c(\gb)f(\gb) \sin(\gs\Theta(\gb))
\end{bmatrix}
$$
$$
M_{22}(\gb;\gs)=\gk\gm\begin{bmatrix}0&f(a)\\ f(\gb)\sin(\gs\Theta(\gb))&f(\gb)\cos(\gs\Theta(\gb))
\end{bmatrix}
$$ for $f(r)\df \frac{\sqrt{\rho(r)}}{r\sqrt{c(r)}},$ as the invertibility of both $M_{11},M_{21}$ is justified by Lemma \ref{proddiff}.
Moreover, with $\gep=\gb-a$ the following hold (mod $\gep^2$ as $\gep\rightarrow 0^+$):
$\det(M_{11}) \equiv-(2\ell +1)a^{-2}\gep;$  $\det(M_{21})\equiv 2\Lambda_{\ell}^2a^{-3}\gep;
$
$$\ad(M_{11})\begin{bmatrix}1\\1\end{bmatrix}=\begin{bmatrix}\Phi_{0,2}'(a)\\\Phi_{0,1}'(a)\end{bmatrix}\gep;
\ad(M_{21})\begin{bmatrix}1\\1\end{bmatrix}=\begin{bmatrix}\Phi_{0,2}''(a)\\\Phi_{0,1}''(a)\end{bmatrix}\gep.$$
 Hence, there are positive constants $K_i$ so that
$$\ad(M_{11}) M_{12}\equiv
K_1\begin{bmatrix}-\Phi_{0,2}^{'}(a)\gep&O(\gep)\\1-\cos(\gs\Theta(\gb))+O(\gep)&\sin(\gs\Theta(\gb))+O(\gep)\end{bmatrix} 
$$
$$\ad(M_{21}) M_{22}\equiv
K_2\begin{bmatrix}-\Phi_{0,2}'(a)\sin(\gs\Theta(\gb))&O(\gep)\\\sin(\gs\Theta(\gb))&\cos(\gs\Theta(\gb))-1\end{bmatrix}
$$

Now, choosing $\gb=\gb_{\gs},$ it follows that $\gs^{-1}M_{21}^{-1}M_{22}=O(1)$ (as $\gs\rightarrow \infty$): Morever,
there are positive constants $\delta_0,\delta_1,\delta_2$ so that 
$$\delta_0z<|\sin(\gs\Theta(\gb_{\gs}))|<\delta_1 z; \,|1-\cos(\gs\Theta(\gb_{\gs}))|<\delta_2 z^2$$ 
uniform for sufficiently large $\gs$ and small $z.$
Then, we have that for 
functions $\ga_i(z)$ and positive constant $\gd$
$$
M_{11}^{-1}M_{12}-\gs^{-1}M_{21}^{-1}M_{22}
= \begin{bmatrix}\ga_1(z)&O(1)\\\ga_3(z)\gs&\ga_2(z)\gs\end{bmatrix}+O(\gs^{-1})
$$  (as $\gs\rightarrow\infty$) 
 where $\ga_3(z)\rightarrow 0$ and  $|\ga_i(z)|$ $>\gd$ for $i=1,2$ as $z\rightarrow 0^+;$ hence, the desired result is now clear.\end{proof}

\begin{prop}\label{lastprop}
For sufficiently large $\gs$ the boundary conditions of the form $\Pi_{\gb_1}(\gs)$ and $\Pi_{j,k;\gb_2}$ 
on $\vec{Y}$ and $\vec{\Phi},$ respectively, are simultaneously admissible, each on any subintervals $\cI_{\gb_1},\cI_{\gb_2}$ $\subseteq$ $[a,b].$
\end{prop} 
\begin{proof} From Corollary \ref{altfundcor} is easy to show that $\fV_{\vec{r}}$ as in (\ref{vs}) satisfies
$$\det\fV_{\vec{r}}(\gs)=\gs^{-1}\frac{c(r_2)f(r_2)}{c(r_1)f(r_1)}|\fw(r_3,r_4)\cos(\gs\Theta_{r_1}(r_2)+\vartheta_1-\vartheta_2)|+O(\gs^{-2}).$$
for $\fw$ as in (\ref{fsfw}). It follows that $\det\fV_{\vec{r}}(\gs)\neq 0$ from Lemma \ref{proddiff} and from choices of $\vartheta_j$ resulting from arguments parallel to those of (\ref{os}) and (\ref{dets}).\end{proof}
\bibliographystyle{plain}

\end{document}